\documentclass{article}

\usepackage{amsmath}
\usepackage{natbib}
\usepackage{soul}
\usepackage{color}
\usepackage{setspace}
\usepackage{amsthm}
\usepackage{amssymb}
\usepackage{geometry}
\usepackage{algorithm}
\usepackage{algpseudocode}
\usepackage{graphicx}
\usepackage{tikz}
\usepackage{authblk}
\usepackage{bm}
\usepackage{bbm}
\usepackage{hyperref}
\usepackage{xcolor}\hypersetup{allcolors=blue,colorlinks=true}

\newcommand{\E}{\mathbf{E\,}}
\renewcommand{\P}{\mathbf{P}}

\newcommand{\R}{\mathbb{R}}
\renewcommand{\H}{\mathcal{H}}
\newcommand{\M}{\mathcal{M}}
\newcommand{\diam}{\textnormal{diam}}
\newcommand{\W}{\mathcal{W}}

\newcommand{\pa}{\mathcal{P}}

\newcommand{\I}{\mathbbm{1}}

\newcommand{\one}[1]{\mathbbm{1}_{\{#1\}}}

\newcommand{\Lip}{\mathrm{Lip}}
\newcommand{\lip}{\mathrm{Lip}}

\newcommand{\cf}[1]{} 

\newcommand{\ccf}[1]{} 

\algrenewcommand\alglinenumber[1]{
    {\sf\footnotesize\addfontfeatures{Colour=888888,Numbers=Monospaced}#1}}
\algrenewcommand\algorithmicrequire{\textbf{Precondition:}}
\algrenewcommand\algorithmicensure{\textbf{Postcondition:}}

\newlength{\zero}
\newtheorem{theorem}{Theorem}
\newtheorem{proposition}{Proposition}
\newtheorem{lemma}{Lemma}

\theoremstyle{definition}
\newtheorem{example}{Example}

\numberwithin{equation}{section}

\title{Space-Filling Design for Nonlinear Models} 
\author[*]{Chang-Han Rhee}
\author[**]{Enlu Zhou}
\author[***]{Peng Qiu}
\affil[*]{Stochastics Group, Centrum Wiskunde \& Informatica\\
    Amsterdam, 1098XG, The Netherlands}
\affil[**]{H. Milton Stewart School of Industrial and Systems Engineering, Georgia Institute of Technology\\
	Atlanta, GA, 30332, USA}
\affil[***]{Wallace H. Coulter Department of Biomedical Engineering, Georgia Institute of Technology and Emory University\\
	Atlanta, GA, 30332, USA}
\date{\today}

\begin{document}

\maketitle

\onehalfspacing
\begin{abstract}
\noindent
Performing a computer experiment can be viewed as observing a mapping between the model parameters and the corresponding model outputs predicted by the computer model. In view of this, experimental design for computer experiments can be thought of as devising a reliable procedure for finding configurations of design points in the parameter space so that their images represent the manifold parametrized by such a mapping (i.e., computer experiments). 
Traditional space-filling design aims to achieve this goal by filling the parameter space with design points that are as ``uniform" as possible in the parameter space. However, the resulting design points may be non-uniform in the model output space and hence fail to provide a reliable representation of the manifold, becoming highly inefficient or even misleading in case the computer experiments are non-linear. 
In this paper, we propose an iterative algorithm that fills in the model output manifold uniformly---rather than the parameter space uniformly---so that one could obtain a reliable understanding of the model behaviors with the minimal number of design points. 


\end{abstract}



%
%
\section{Introduction}\label{sec:intro}
By virtue of the immense computational capacity of modern computing machineries, experimentation via computer simulation 
has become an integral element of science and engineering. 
Due to its deterministic nature, however, designing computer simulation experiments requires a different approach from traditional design of physical experiments.
A useful perspective is to view the given computer experiment as a function that maps a set of input variables to output variables (Figure~\ref{fig:1}).
In view of this, one can regard a design of computer experiment as an exploration of a manifold parametrized by such a function. 
Interesting design questions can be answered by sampling from a given distribution on such a manifold.
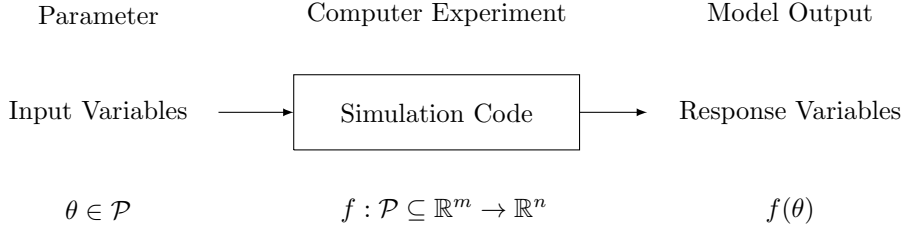
\begin{figure}\label{fig:1}
\center
\begin{tikzpicture}
\draw (-3,1.8) node {Parameter};
\draw (6.2,1.8) node {Model Output};
\draw (1.5,1.8) node {Computer Experiment};
\draw (-3,0.5) node {Input Variables};
\draw (6.2,0.5) node {Response Variables};
\draw[-latex] (-1.4,0.5)--(-0.4,0.5);
\draw[-latex] (3.4,0.5)--(4.3,0.5);
\draw (-0.4,0) to (-0.4,1) to (3.4,1) to (3.4,0) to (-0.4,0);
\draw (1.5,0.5) node {Simulation Code};
\draw (-3,-0.8) node {$\theta\in \pa$};
\draw (1.6,-0.8) node {$f:\pa\subseteq \R^m\to \R^n$};
\draw (6.2,-0.8) node {$f(\theta)$};
\end{tikzpicture}
\caption{Computer experiments can be viewed as a function that maps input variables (parameters) to response variables (model outputs).}
\end{figure}

Suppose that $f:\pa\subset\R^m \to \R^n$ is a mapping on a hypercube $\pa = \prod_{i=1}^m[x_{\min}^i, x_{\max}^i]\subseteq \R^m$; we assume that $f$ possesses sufficient regularity so that the image $f(\pa)$ is an $m$-dimensional manifold embedded in $\R^n$. 
Let $\M$ denote this manifold.
Throughout this paper, we will call $\mathcal P$ as the parameter space or input space, and $\mathcal M$ as the manifold or output space.
The goal of this paper is to develop a reliable and efficient computational procedure that finds configurations $\{x_1,\ldots,x_n\}\subseteq A$ in the parameter space in such a way that $\{f(x_1),\ldots,f(x_n)\}$ represent $\mathcal M$ well. 
We will make it clear what we mean by this in rigorous mathematical statements in Section~\ref{sec:algorithm}, but here we first explain the challenges and the objectives at an intuitive level through an illustrative example.
Consider a mapping 
$f:\pa\subseteq \R^2 \to \R^3$
$$
f(\theta_1,\theta_2) 
\triangleq
\left(
\begin{array}{c}
e^{-\theta_1 t_1} + e^{-\theta_2 t_1}\\
e^{-\theta_1 t_2} + e^{-\theta_2 t_2}\\
e^{-\theta_1 t_3} + e^{-\theta_2 t_3}
\end{array}
\right)
$$
on $\pa = [0,100]^2$ and the associated manifold 
\begin{equation}\label{eq:manifold-exponential-model}
\mathcal M = \{(e^{-\theta_1 t_1} + e^{-\theta_2 t_1},e^{-\theta_1 t_2} + e^{-\theta_2 t_2}, e^{-\theta_1 t_3} + e^{-\theta_2 t_3}): \theta_1,\theta_2 \in[0, 100]\}
\end{equation}
where $t_1=1$, $t_2=2$, $t_3=4$. 
Although this example is given in a closed-form formula for the purpose of illustration, a typical situation we would like to address in this paper is when the evaluation of $f$ is only possible through an expensive black box simulation.
Think of $f$ as a model describing the dynamics of the system so that $f(\theta_1,\theta_2)$ is the model output given the parameter $(\theta_1,\theta_2)$. 
$\mathcal M$ then can be viewed as all of the possible behaviors of the system from the parameters in $\pa$. 
Perhaps, the most straightforward way to investigate the model behavior---e.g., the shape and range of the manifold $\mathcal M$, the sensitivity of the model output in the input parameters, etc.---is to evaluate the mapping $f$ at a large enough number of different parameters, say, $\{x_1,\ldots,x_n\} \subseteq \pa$ so that they ``cover'' $\pa$ sufficiently well and then observe $\{f(x_1),\ldots,f(x_n)\}$. 
Traditional space-filling designs carefully construct design points $\{x_1,\ldots,x_n\}$ in such a way that the parameter space $\pa$ is ``well covered'' by $\{x_1,\ldots,x_n\}$; see for example, \cite{santner2013design}.
While such a strategy can be a powerful means to study the manifolds, it should be noted that if $f$ is parametrized in a highly nonlinear way so that the vast majority of the parameter space $\pa$ is mapped into a small part of the manifold $\mathcal M$, then na\"ive choices of configuration $\{x_1,\ldots,x_n\}$ can lead to not only very inefficient computational procedures but also dangerously misleading observations. 
The manifold $\mathcal M$ in \eqref{eq:manifold-exponential-model} illustrates this point clearly. 
Figure~\ref{fig:expo} displays the samples generated on $\mathcal M$ in two different ways. 
The left plot displays  $5,000$ samples $f(x_1), \ldots, f(x_{5,000})$ where the $x_i$-s are generated uniformly on $\pa$ according to the traditional principle, 
while the right plot displays $5,000$ samples $f(x'_1),\ldots,f(x'_{5,000})$ where $x'_i$-s are generated so that $f(x'_i)$-s are uniformly distributed on $\mathcal M$. 
One can see that most of the $x_i$-s are mapped into a small fraction of $\mathcal M$ in the left plot, and hence, $\{f(x_1),\ldots,f(x_n)\}$ do not reflect the actual geometry of $\mathcal M$ in a reliable way. 
If the output behavior that the na\"ive design points failed to capture in the left plot are undesirable behaviors (that the experimenter wants to avoid), then the experimenter may incorrectly conclude that the model is robust to the perturbation of the parameters and get a false sense of safety. 
On the other hand, if the missed output behaviors are desirable ones,  then the  experimenter will miss opportunities to discover and take advantage of the such model behaviors.  
This illustrates the potential danger in drawing conclusions from the observations based on blindly choosing uniform configurations in the parameter space without considering the behavior of the model $f$.
\begin{figure}[htb]
\centering
\includegraphics[width = \linewidth]{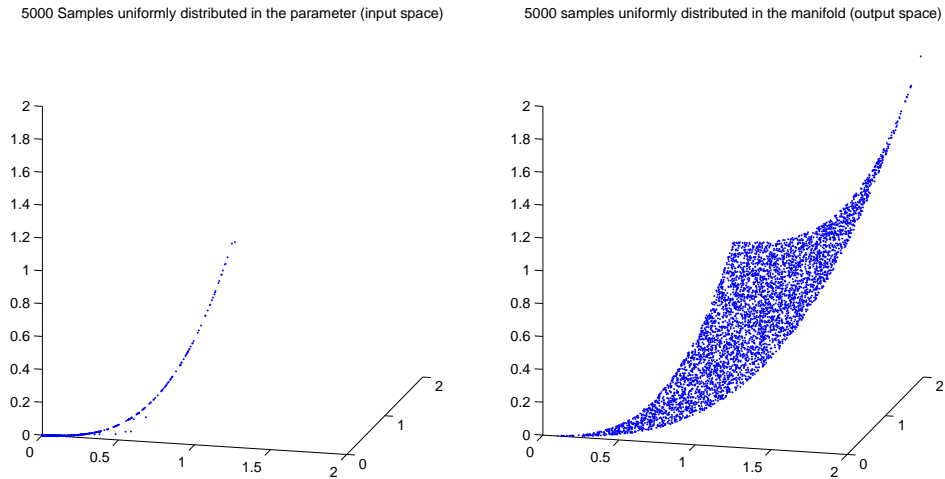}
\caption{$5,000$ samples generated uniformly on the parameter space $\pa$ (left plot) vs.\ uniformly on the manifold $\mathcal M$ (right plot). 
Observations based on the left plot may lead to incorrect inferences.
For many purposes, more reliable and informative  configurations for such nonlinear models would be uniform (or other desired) distribution on $\mathcal M$ rather than $\pa$. 
\label{fig:expo}}
\end{figure}

In this paper, we propose computational procedures for generating configurations (design points) $\{x_1,\ldots,x_n\}$ on $\pa$ so that the resulting configuration $\{f(x_1),\ldots,f(x_n)\}$ on $\mathcal M$ is uniformly distributed (or according to other desired distributions) as in the right plot of Figure~\ref{fig:expo}. 
The idea is to start with a random configuration, and then shift the configuration towards the target distribution by alternating between resampling in the model output space and perturbation in the parameter space. 
In the resampling step, the algorithm resamples the points in such a way that the points that belong to a densely populated region of $\mathcal M$ are less likely to be resampled, while the points that belong to the sparse region of $\mathcal M$ are more likely to be resampled so that the resulting samples are populated more uniformly. 
While this step pushes the empirical distribution of the design points toward the target distribution, it cannot be repeated to push the design points further toward the target distribution since no new points in $\mathcal M$ are discovered while some are discarded. 
To discover new regions of $\mathcal M$, the resampling step is followed by the perturbation step where each points are shifted around in such a way that the distribution of the shifted points are not too far away from the original distribution. 
The two steps are repeated until the configuration is sufficiently uniform on $\mathcal M$.
We study the consistency and the convergence of the proposed algorithm in Kantorovich-Rubinstein distance. 
A numerical investigation (without convergence analysis) of the preliminary version of one of the algorithms discussed in this paper has been presented in \cite{rhee2014iterative}.

It should be pointed out that our setting is different from the setting where algorithms were developed to generate uniform samples on manifolds based on random walks, such as hit-and-run \cite{boneh1979constraints,smith1984efficient} and shake-and-bake \cite{boender1991shake}, 
stochastic billiard \cite{dieker2015stochastic},
geodesic-walks \cite{lee2016geodesic}. 
Such algorithms work when the manifold is a convex set or the boundary of a convex set, and the manifold is specified by a set of constraints (eg. a polytope given by a system of linear inequalities), and hence, it is easy to tell whether a given point in $\R^n$ (the ambient space in which the manifold is embedded) belongs to the manifold or not.
In contrast, in our setting, the manifold is not necessarily a convex set or the boundary of a convex set, and more importantly, we do not have a direct way to answer whether a given point in $\R^n$ belongs to the manifold or not.

The rest of the paper is organized as follows.
Section~\ref{sec:algorithm} formulates the objectives of this paper in rigorous mathematical formulation and presents our main algorithm.
Section~\ref{sec:consistency-convergence-analysis} analyzes the consistency and the convergence of the proposed algorithm.
Section~\ref{sec:proofs} provides the technical proofs of the results in Section~\ref{sec:consistency-convergence-analysis}. 
Section~\ref{sec:examples} examines the numerical behavior of our algorithm with a few illustrating examples including a systems biology model for enzymatic reaction networks.

%
%
\section{Mathematical Formulation and Algorithm Description}\label{sec:algorithm}
Section~\ref{subsec:mathematical-formulation} recalls the preliminary mathematical notions required to state the objectives of this paper, and Section~\ref{subsec:algorithm-description} proposes the algorithm that achieves the objectives stated in Section~\ref{subsec:mathematical-formulation}.

\subsection{Mathematical Formulation}\label{subsec:mathematical-formulation}
Let $f:\pa\subset\R^m \to \R^n$ be a mapping on a hypercube $\pa \triangleq \prod_{i=1}^m[x_{\min}^i, x_{\max}^i]\subseteq \R^m$ with sufficient regularity so that $\mathcal M \triangleq f(\pa)$ is the associated $m$-dimensional manifold embedded in $\R^n$. 
We assume that $f$ is not analytically tractable but can be evaluated by a simulation code at arbitrary points in $\pa$. 
Such evaluation often requires significant computational resource.
Further, we will assume for simplicity of the discussion that $f$ is one-to-one.
Our goal is to find a configuration $\{x_1,\ldots,x_n\}\subseteq \pa$, for which the pairs $(x_1,f(x_1)),\ldots,(x_n,f(x_n))$ provide reliable information regarding $\mathcal M$.
In view of the discussion in Section~\ref{sec:intro}, we wish to generate samples  uniformly (in some sense) on the manifold $\mathcal M$. 
A natural notion of uniform distribution on a manifold can be stated in terms of the Hausdorff measure. 
Recall that the $m$-dimensional Hausdorff measure is defined as
\begin{equation}
\mathcal H^m (B) = \lim_{\delta\to 0} \inf_{\substack{B\subseteq \cup S_i,\\\text{diam}(S_i)\leq \delta}}\sum \Gamma_m \left(\frac{\text{diam}(S_i)}{2}\right)^m
\end{equation}
where the infimum is over all countable coverings $\{S_i\subseteq \R^n:i\in I\}$ of $B$, $\text{diam}(S_i) \triangleq \sup\{|x-y|: x,y\in S_i\}$, and $\Gamma_m$ is the volume of the $m$-dimensional unit ball.
Note that $\mathcal H_m$ is a natural generalization of Lebesgue measure, and if $m = n$, $\mathcal H^m$ coincides with the $m$-dimensional Lebesgue measure; see, for example,  \cite{Federer1996} for more details.

As \cite{Diaconis2013} points out, the area formula (see, for example, Section 3.2.5 of \citealt{Federer1996}) in geometric measure theory dictates how one should sample from a given density with respect to the Hausdorff measure.
\begin{proposition} (Area Formula, \citealt{Federer1996}) If $f:\R^m \to \R^n$ is Lipschitz and $m\leq n$, for any measurable $A$ and measurable $g:\R^n\to\R$,
\begin{equation}\label{eq:area_formula}
\int_A g(f(x))J_mf(x) \lambda^m(dx)  = \int_{\R^n} g(y)(\#\{A\cap f^{-1}(y)\})\ \mathcal H^m(dy)
\end{equation}
where $\lambda^m$ denotes the $m$-dimensional Lebesgue measure, $\#S$ denotes the cardinality of the set $S$, and $J_k f$ is the $k$-dimensional Jacobian of $f$. In this special case where $k = m\leq n$, the $k$-dimensional Jacobian is equal to
$$J_mf(x) = \sqrt{\textup{det}\big(Df(x)^TDf(x)\big)},$$
where $Df(x)$ is the differential of $f$ at $x$
$$
Df=
\left(
\begin{array}{cccc}
\frac{\partial}{\partial x_1} f_1
&\frac{\partial}{\partial x_2} f_1
&\ldots
&\frac{\partial}{\partial x_m} f_1
\\
\frac{\partial}{\partial x_1} f_2
&\frac{\partial}{\partial x_2} f_2
&\ldots
&\frac{\partial}{\partial x_m} f_2
\\
\vdots
&\vdots
&\ddots
&\vdots
\\
\frac{\partial}{\partial x_1} f_n
&\frac{\partial}{\partial x_2} f_n
&\ldots
&\frac{\partial}{\partial x_m} f_n
\end{array}
\right)
.
$$
\end{proposition}
Now our goal can be formulated as sampling from a given distribution $\eta(\cdot)$ supported on the manifold $\M$. 
In particular, we assume that $\eta(\cdot)$ is absolutely continuous with respect to the Hausdorff measure on $\M$, and the density (i.e., Radon-Nikodym derivative) is $\mu$, which is known up to a multiplicative constant. 
In view of the area formula (\ref{eq:area_formula}), to generate samples from $\eta$, one can generate samples $x_1, x_2, \ldots$ in the parameter space from the density proportional to $\mu\circ f\cdot J_mf$ and then apply $f$ to $x_1, x_2, \ldots$ to obtain the samples $f(x_1), f(x_2), \ldots$ from the desired density on the manifold.
To see why such a procedure generates samples from the desired distribution, pick $g(x) = \mu(x) I_B(x)$  for a given set $B$. 
Then from \eqref{eq:area_formula},
$P(X \in f^{-1}(B)) = \int \mu\circ f (x) J_m f(x) I_{f^{-1}(B)}(x) dx = \int \mu(y) I_{B}(y) \mathcal H^m(dy) = P( Y \in B)$, where  $Y$ is an $\mathcal M$-valued random variable with the density $\mu$, and $X$ is an $\pa$-valued random variable with density $\mu\circ f \cdot J_m f$. 
Since $B$ was chosen arbitrarily, we see that $f(X)$ and $Y$ have the same distribution. 
That is, if we sample $X$ in $\pa$ from the density  $\mu\circ f \cdot J_mf$ (w.r.t. the Lebesgue measure), then $f(X)$ is a random variable in $\mathcal M$ with density $\mu$ (w.r.t. the Hausdorff measure $\mathcal H^m$).

Sampling from a given density with respect to the Lebesgue measure is a classical topic that has been  addressed by many traditional methods such as inversion, acceptance-rejection, and Markov chain Monte Carlo;
in particular, Markov chain Monte Carlo algorithms such as Metropolis-Hastings provide powerful means to sample from analytically intractable densities; see for example \cite{asmussen2007stochastic,liu2008monte,robert2004monte}. 
However, it should be noted that our goal is different from the context where Markov chain Monte Carlo methods are typically deployed. 
Markov chain Monte Carlo algorithms produce samples that conform with the target distribution by rejecting many proposals that do not conform with the target distribution. 
Deciding whether or not to reject the proposal requires computation of likelihoods, which corresponds to performing computer experiments in our context.
That is, if we use Markov chain Monte Carlo algorithms for our purpose, many computer experiments will have to be performed at the points that are not ``right".
Considering that our goal is to find minimal design points that produce a robust representation of the manifold, this feature of Markov chain Monte Carlo approach defeats the purpose.
Moreover, Markov chain Monte Carlo algorithms are typically designed to compute the integral of certain functions w.r.t.\ the given distribution, and hence, the conventional performance criteria are concerned with the mean square error of such integrals. 
On the other hand, our goal is not merely computing integrals w.r.t.\ the target measure.
Rather, our goal is to generate samples whose images cover the whole manifold $\mathcal M$ so that they represent the manifold in reliable ways.   
Among the existing sampling techniques, what comes closest to our spirit is perhaps non-parametric importance sampling: \cite{zhang1996nonparametric,givens1996local,kim2000nonparametric,zlochin2002efficient}.
In fact, if the evaluation of the derivative of $f$ is also available in addition to the evaluation of $f$ itself, our algorithm can be simplified to a version which can be seen as a variant of non-parametric importance sampling algorithms. 
However, previous studies of such algorithms have been focused on computing a single test function by approximating the zero-variance importance sampling measure with kernel density proposal. 
In view of our purpose, a distance between the target measure and the empirical measure of the samples produced by the algorithm would be a more proper perfomance measure. 
We analyze the convergence of our algorithm with respect to the Kantorovich-Rubinstein distance (which is also know as Wasserstein distance of order $1$).
To the best of the authors' knowledge, the convergence bound we establish for our algorithm in this paper is the first convergence analysis of non-parametric importance sampling type algorithms in terms of Kantorovich-Rubinstein distance.




We conclude this section with a brief review of Kantorovich-Rubinstein distance.
Kantorovich-Rubinstein distance is the $L^1$ distance between the optimal coupling of the two random variables whose marginal distributions coincide with the probability distributions. 
That is, 
$$\W_1(\mu, \nu) \triangleq \inf_{\pi\in M(\mu, \nu)} \int_{A\times A} \|x-y\|_1 d\pi(x,y) $$
where $M(\mu,\nu)$ denotes the set of all joint probability measures on $A\times A$ with marginals $\mu$ and $\nu$ respectively. Obviously, this is equivalent to
$$\W_1(\mu, \nu) = \inf_{X,Y} \E \| X - Y\|_1$$
where the infimum is taken over all coupling of $\mu$ and $\nu$. 
The following dual formula will be useful in the analysis of the modulus of continuity of the resampling step:
$$\W_1(\mu, \nu) = \sup \left\{\int_A f(x) \mu(dx) - \int_A f(x) \nu(dx): \|f\|_\lip \leq 1\right\}$$
where $\|f\|_\lip$ denotes the Lipschitz constant of $f$. 
Convergence in Kantorovich-Rubinstein distance implies weak convergence.
See, for example, Chapter 6 of \cite{villani2008optimal} for more details.

\subsection{Algorithm Description}\label{subsec:algorithm-description}

In this section, we propose an algorithm that generates a sequence of design points  whose empirical distribution converges to the target distribution $\mu$ in Kantorovich-Rubinstein distance.




The main idea of our algorithm is to start with arbitrarily distributed samples and then repeat iterations consisting of a resampling step and a perturbation step so that the empirical distribution of the samples becomes closer and closer to the target distribution. 
The resampling step is designed to shift the current empirical distribution toward the target distribution by eliminating the samples in concentrated regions, and duplicating the samples in  sparse regions; the perturbation step is designed to force the algorithm to explore new areas of the manifold while still respecting the information obtained through the previous iterations. 
More specifically, the algorithm works as follows.
At iteration 0, one starts with $N$ (arbitrarily chosen) points $x_1,\ldots,x_N$ and their images $y_1,\ldots,y_N$ where $y_i \triangleq f(x_i)$ for $i=1,\ldots,N$.
Let $B(y_i;r)$ denote the $n$-dimensional ball with radius $r$ centered at $y_i$, 
and $\hat r(y_i)$ denote the $k$\textsuperscript{th} nearest neighborhood distance from $y_i$. 
That is, 
$\hat r(y_i) = \hat r \circ f(x_i)\triangleq \inf\{r>0: \#\{j: y_j \in B(y_i;r)\} \geq k\}$.
For each $j>0$, the $j$\textsuperscript{th} iteration consists of a resampling step and a perturbation step. 
In the resampling step, one computes the resampling weights $G_i$ as follows:
\begin{equation}\label{eq:resampling-formula}
G_i \triangleq \frac{ \hat r^m\circ f(x_i)\cdot (\mu \circ f)( x_i) }{ \sum_{l=1}^N \hat r^m\circ f(x_l) \cdot (\mu \circ f)( x_l)},\qquad \forall i=1,\ldots,N
\end{equation}
where $\hat r^m(\cdot)$ denotes the $m$\textsuperscript{th} power of $\hat r(\cdot)$---i.e., $\hat r^m(y) = (\hat r(y))^m$ for each $y\in \R^n$. 

Then, the algorithm generates iid samples $x_1',\ldots,x_N'$ in such a way that $\P(x_i' = x_j) = G_j$ for each $i, j=1,\ldots, N$.
For the perturbation step, fix constants (algorithmic parameters) $q\in(0,1)$, $b\gg 1$, and $h>0$ as well as a scaled kernel $\tilde \zeta_h(\cdot;y)$ centered at $y$ with bandwidth $h$. 
The constants $q$ and $b$ regularize the perturbation density so that the density is bounded away from 0 and $\infty$, while $h$ is the perturbation bandwidth. 
The choice of these parameters and the precise construction of $\tilde \zeta_h$ will be discussed further in Section~\ref{sec:consistency-convergence-analysis} and Section~\ref{sec:examples}.
One starts the perturbation step with the samples $x_1',\ldots, x_N'$ generated in the previous resampling step, and constructs a smoothed and regularized density
\begin{equation}\label{eq:regularizedDensity}
\frac{\min\left\{b, \ q/\lambda^m(\pa) + (1-q)\frac{1}{N}\sum_{i=1}^{N}\tilde\zeta_h(x;x_i')\right\}}
{\int_A \min\left\{b, \ q/\lambda^m(\pa) + (1-q)\frac{1}{N}\sum_{i=1}^{N}\tilde\zeta_h(s;x_i')\right\}ds}
\end{equation}
for some sufficiently large $b$. 
For each $i$, generate $x_i$ from \eqref{eq:regularizedDensity}.
To generate a sample from this density, the algorithm generates a proposal $x^*$ uniformly (on $\pa$) with probability $q$, and according to $\frac1N\sum_{i=1}^N \tilde\zeta_h(x;x_i')$ with probability $1-q$. 
Set $a = q/\lambda^m(\pa) + (1-q)\frac{1}{N}\sum_{i=1}^{N}\tilde\zeta_h(x;x_i')$. 
Accept $x^*$ (i.e., set $x_i = x^*$) with probability $\max\{a,b\}/a$. 
If not accepted, reject $x^*$ and repeat until some proposal is accepted so that $x_i$ is set. 
When $x_1,\ldots,x_N$ are all generated from this procedure, one moves on to the resampling step of the next iteration. 
The whole procedure described so far is summarized in Algorithm~\ref{alg:iterative}.

\begin{algorithm}
	\caption{Space-Filling Algorithm (without Derivative)}
	\label{alg:iterative}
	\begin{algorithmic}
	\vspace{1pt}
	\State Generate $N$ samples $x_1, \cdots, x_{N} \in \pa$ from an initial distribution $p_0$;
	\While {$\hat\eta$ changes notably}
		\State $\{x_1',\cdots,x_{N}'\} \gets \textsc{Resample}(\{x_1,\cdots,x_{N}\})$;
		\State $\{x_1,\cdots,x_{N}\} \gets \textsc{Perturb}(\{x_1',\cdots,x_{N}'\})$;
		\State $\hat\eta \gets \frac1{N} \sum_{i=1}^{N} \delta_{f(x_i)}$;
	\EndWhile
	\State\Return $\hat\eta$\\

	\Function{Resample}{$\{x_1,\cdots,x_N\}$}
		\State $y_i \gets f(x_i)$,\qquad\qquad\qquad\qquad\qquad\qquad\ \ \ \ \,$i=1,\ldots,N$;\hfill
 		\State $\hat r_i \gets \text{$k$-NN distance from $y_i$}$,\qquad\qquad\qquad$i=1,\ldots,N$;\hfill 
		\State $G_i \gets {\hat r_{i}^m\cdot\mu(y_i)}/({\sum_{l=1}^n \hat r_{l}^m\cdot\mu(y_l)})$,\hspace{33.5pt}$i=1,\ldots,N$;\hfill 
		\For {$i=1:N$}
			\State Sample $x_i'$ so that $\P(x_i' = x_l) = G_l$, $\forall l=1,\ldots,N$;
		\EndFor
		\State \Return $\{x_1',\cdots,x_{N}'\}$;
	\EndFunction\\

	\Function{Perturb}{$\{x_1',\cdots,x_N'\}$}
		\For {$i=1:N$}
			\While {$x^*$ is not accepted}
					\State Draw $x^*$ from the density $q/\lambda^m(\pa) + (1-q)\frac{1}{n}\sum_{l=1}^n\tilde\zeta_h(x;x'_l)$;
				\State $a \gets q/\lambda^m(\pa) + (1-q)\frac{1}{n}\sum_{l=1}^n\tilde\zeta_h(x^*;x'_l)$;
				\State Accept $x^*$ and set $x_i=x^*$ with probability $\frac{\min\left\{a,b\right\}}{a}$;
			\EndWhile
		\EndFor
		\State \Return $\{x_1,\cdots,x_N\}$;
	\EndFunction\\
	\end{algorithmic}
\end{algorithm}
In the rest of this section, we provide some intuition behind the design and analysis of Algorithm~\ref{alg:iterative} and propose a simplified version Algorithm~\ref{alg:ideal} in case the derivative of $f$ can be evaluated in addition to the value of $f$ itself. 
Roughly speaking, $1/\hat r^m$ is approximately proportional to the density $p_Y$ of the current samples, and hence, the resampling formula assigns probability mass approximately proportional to the likelihood ratio $\mu/p_Y$ between the target density and the density of the current samples. 
In view of this, the resulting empirical distribution of the samples after the resampling step should be closer to the target distribution.
To be more specific, the resampling formula \eqref{eq:resampling-formula} can be understood as an approximate Boltzmann-Gibbs transformation:
recall that the Boltzmann-Gibbs transformation $\Psi_G(\eta)$ of $\eta$ w.r.t.\ the potential $G$ is defined as a  measure such that 
$$
\Psi_G(\eta)(dy) \triangleq \frac{G(y)\eta(dy)}{\int G(y)\eta(dx)}.
$$
Suppose that $x_1,\ldots,x_N$ are the samples generated by the perturbation step in the previous iteration  and we integrate a test function $g$ w.r.t.\ the weighted empirical measure $\hat\eta = \frac 1N \sum_{i=1}^N G_i\delta_{f(x_i)}$ and compare it to the integral of $g$ w.r.t.\ the target measure $\eta(dy) =\mu(y)\H^m(dy)$: 
\begin{equation}\label{eq:the_difference}
\sum_{i=1}^{N}\frac{\hat r^m(y_i)  \mu(y_i)}
{\sum_{j=1}^{N}\hat r^m(y_j) \mu(y_j) }g(y_i)
- \int_\M g(y)\eta(dy).
\end{equation}
Set the potential $\hat G_Y$ as
$$ \hat G_Y(y) \triangleq \mu(y)\Gamma_m \hat r^m(y)/(k/N)\triangleq \mu(y)/\hat p_Y(y)$$
where $\hat p_Y$ can be viewed as the $k$-nearest neighbor approximate density of $y_i$'s
and denote the empirical distribution of $y_i$'s with $\hat\eta_Y = \frac{1}{N}\sum_{i=1}^N\delta_{y_i}$.
Note that the first term in (\ref{eq:the_difference}) can be viewed as the expectation of $g$ with respect to the measure obtained by applying the Boltzmann-Gibbs transformation to $\hat\eta_Y$ with respect to the potential $\hat G_Y$. 
That is, $\hat \eta = \Psi_{\hat G_Y}(\hat\eta_Y)$.
Set the potential $G_Y$ to be
$$G_Y(y) \triangleq \lim_{N,k\to \infty,\ k/N \to 0} \mu(y)\Gamma_m r_{k,N}(y)^m/(k/N)=\mu(y)/p_Y(y),$$
where $r_{k,N}(\cdot)$ is the diameter of the ball that contains the $k/N$ fraction of $y_i$'s distribution, i.e., for each $z$,
$$\int_{B(z; r_{k,N}(z))} p_Y(y) \H^m(dy)  = k/N, $$
and $p_Y$ is the density of $y_i$'s.
Note that $\hat r$ approximates $r_{k,N}$. 
The second term in (\ref{eq:the_difference}) can be rewritten as 
$$\int_\M g(y)\frac{\mu(y)}{p_Y(y)}p_Y(y)\H^m(dy),$$ 
and hence, can be interpreted as the expectation of $g$ with respect to the measure obtained by applying the Boltzmann-Gibbs transformation to $\eta_Y(dy) \triangleq p_Y(y)\H^m(dy)$ with respect to the potential $G_Y$.
Therefore, (\ref{eq:the_difference}) can be re-written as
\begin{align*}
\hat \eta g - \eta g
=
\Psi_{\hat G_Y}(\hat\eta_Y)g - \Psi_{G_Y}(\eta_Y)g.
\end{align*}
In view of this, if $\hat\eta_Y$ is a good approximation of $\eta_Y$ and $\hat G_Y$ is a good approximation of $G_Y$, we expect that the difference will be small.
Since $g$ is an arbitrary test function, this suggests that the (weighted) empirical measure $\hat \eta$ resulting from the resampling formula should be a good approximation of $\eta$, which explains, at an intuitive level, why the resampling formula works. 
Finally, note that if we let $\xi$ be the probability measure (on $\pa\subset\R^m$) with density $\mu\circ f \cdot J_m f$, i.e., 
\begin{equation}\label{eq:density_of_xi}
\frac{d\xi}{d\lambda^m}(x) = \left(\mu\circ f(x)\right)\cdot \left(J_m f(x)\right),
\end{equation} 
then $\xi$ is the target measure on the parameter space. 
In other words, if $X\sim \xi$, then $f(X) \sim \eta$. 
It should also be noted that $(\hat r^m\circ f\cdot \mu\circ f)$ in \eqref{eq:resampling-formula} can be regarded as an (normalized) approximation of $(\mu\circ f\cdot J_m f \cdot \iota )$ where $1/\iota$ is the density from which  $x_i$-s are sampled from. 
In view of this, in case one can readily evaluate $J_mf$, a simplified version of Algorithm~\ref{alg:iterative} can be implemented by replacing $(\hat r^m\circ f\cdot \mu\circ f)$ in \eqref{alg:iterative} with $(\mu\circ f\cdot J_m f \cdot \iota )$. 
In this case, the perturbation step can also be simplified; it is not necessary to truncate the sampling density in the perturbation step at level $b$. 
Such a simplified version of the algorithm is summarized in Algorithm~\ref{alg:ideal}.

\begin{algorithm}
	\caption{Space-Filling Algorithm (with Derivative)}
	\label{alg:ideal}
	\begin{algorithmic}
	\vspace{1pt}
	\State Generate $N$ samples $x_1, \cdots, x_{N} \in \pa$ from an initial distribution $p_0$;
	\While {$\hat\eta$ changes notably}
		\State $\{x_1',\cdots,x_{N}'\} \gets \textsc{Resample}(\{x_1,\cdots,x_{N}\})$;
		\State $\{x_1,\cdots,x_{N}\} \gets \textsc{Perturb}(\{x_1',\cdots,x_{N}'\})$;
		\State $\hat\eta \gets \frac1{N} \sum_{i=1}^{N} \delta_{f(x_i)}$; 
	\EndWhile
	\State\Return $\hat\eta$\\

	\Function{Resample}{$\{x_1,\cdots,x_N\}$}
		\State $\mu_i \gets \mu\circ f(x_i)$ \qquad\qquad\qquad\qquad\qquad\qquad\qquad \ \ \ \,$i=1,\ldots,N$;\hfill
		\State $J_i \gets J_mf(x_i)$\qquad\hspace{137.5pt}$i=1,\ldots,N$;\hfill
		\State $\iota_i \gets \big(q/\lambda^m(\pa) + (1-q)\frac{1}{N}\sum_{l=1}^N\tilde\zeta_h(x_i;x'_l)\big)^{-1}$\hspace{15pt} $i=1,\ldots,N$;\hfill
		\State $G_i \gets {\mu_i\cdot J_i\cdot \iota_i }/({\sum_{l=1}^n \mu_i\cdot J_i\cdot \iota_l})$\qquad \hspace{54pt}$i=1,\ldots,N$;\hfill 
		\For {$i=1:N$}
			\State Sample $x_i'$ so that $\P(x_i' = x_l) = G_l$, $\forall l=1,\ldots,N$;
		\EndFor
		\State \Return $\{x_1',\cdots,x_{N}'\}$;
	\EndFunction\\

	\Function{Perturb}{$\{x_1',\cdots,x_N'\}$}
		\For {$i=1:N$}
			\State 
			Draw $x^*$ from the density $q/\lambda^m(\pa) + (1-q)\frac{1}{n}\sum_{l=1}^n\tilde\zeta_h(x;x'_l)$;
		\EndFor
		\State \Return $\{x_1,\cdots,x_N\}$;
	\EndFunction\\
	\end{algorithmic}
\end{algorithm}

%
%
\section{Consistency and Convergence Analysis}\label{sec:consistency-convergence-analysis}

In this section, we provide sufficient conditions for the convergence of the proposed algorithms.
We make the following assumptions: 
\begin{itemize}
\item[A1.] $f$ is $C^2_b(\pa)$;
\item[A2.] $\partial_i J_mf$ vanishes on the boundary $\partial \pa$ of $\pa$ for each $i=1,\ldots, m$.
\end{itemize}
The above conditions are imposed for the purpose of facilitating the convergence proof. 
We expect that Algorithm~\ref{alg:iterative} and Algorithm~\ref{alg:ideal} are consistent under much more general conditions. 
A2 simplifies the analysis since it allows $\tilde \zeta_h$ to produce consistent density estimation in terms of the derivative of the density in addition to the value of the density itself at the boundary; see (iii) of Lemma~\ref{lemma:new-l_infty_proximity}. 
If A2 does not hold, one can still prove the consistency by carefully dealing with the boundary of $\pa$ separately as in Appendix~\ref{sec:knn-proofs}. 
Our goal in this section is to show that the empirical distribution of the points $\{x_1,\ldots,x_{N}\}$ produced by Algorithm~\ref{alg:ideal} ``converges'' to the target distribution $\xi$ in $\W_1$ as the number of samples $N$ and the iterations $j$ grow. 
The precise statement will be given in Theorem~\ref{thm:main-result}.
The consistency analysis of Algorithm~\ref{alg:iterative} is more involved. 
We provide a discussion of Algorithm~\ref{alg:iterative} in Appendix~\ref{sec:knn-proofs}.

Before starting the analysis, we construct a kernel that is consistent near the boundary.
In simple words, we are defining these kernels with reflection along the boundaries. 
That is, if $\pa = [0,1]$, we fill in $[-1,0]$ with the mirror images of the points on $[0,1]$ (axis of reflection: $x=0$) and also fill in $[1,2]$ with the mirror images of the points on $[0,1]$ (axis of reflection: $x=1$) so that the data does not abruptly disappear outside of the boundary. 
This eliminates the boundary effect in the sense that the resulting kernel density estimation is consistent on the boundary; see (i) of Lemma~\ref{lemma:new-l_infty_proximity}. 
To be specific, consider a smooth and symmetric kernel $\zeta$ supported on $B(0;1)$, such as biweight kernel, triweight kernel, and tricube kernel, and
set
$$\tilde \zeta_h(x;y) \triangleq \tilde\zeta_h^{(m)}(x;y)$$
where 
$\tilde\zeta_h^{(0)}(x;y) \triangleq \zeta_h(x-y) \triangleq \frac{1}{h^m}\zeta\left(\frac{x-y}{h}\right)$
and $\tilde \zeta_h^{(i)}$ are defined recursively for $i=1,\ldots,m$ as
$$
\tilde \zeta_h^{(i+1)}\left(x;y\right) = \tilde \zeta_h^{(i)}(x;y) + \tilde\zeta_h^{(i)}\big(x;\,\text{refl}_{\min}(y;i+1)\big) + \tilde\zeta_h^{(i)}\big(x;\,\text{refl}_{\max}(y;i+1)\big),
$$
where $x = (x^1,\ldots,x^m)^T$, $y = (y^1,\ldots,y^m)^T$,
and 
$$
\text{refl}_{\min}(y;i) 
\triangleq
\left(\begin{array}{l}\phantom{2x_i^{\min}-\ }y^1\\\phantom{2x_i^{\min}-}\ \vdots\\\phantom{2x_i^{\min}-\ }y^{i-1}\\2x^i_{\min}-y^i\\\phantom{2x_i^{\min}-\ }y^{i+1}\\\phantom{2x_i^{\min}-}\ \vdots\\\phantom{2x_i^{\min}-\ }y^m\end{array}\right)
\qquad\text{and} \qquad
\text{refl}_{\max}(y;i) 
\triangleq
\left(\begin{array}{l}\phantom{-}y^1\\\phantom{-}\ \vdots\\\phantom{-}y^{i-1}\\-y^i+2x^i_{\max}\\\phantom{-}y^{i+1}\\\phantom{-}\ \vdots\\\phantom{-}y^m\end{array}\right).
$$
Note that for $x,y$ such that $|x-y|>h$,
\begin{equation}\label{zeta-is-zero}
\tilde \zeta_h(x;y) = 0,
\end{equation}
and for any $x\in \pa$,
\begin{equation}\label{zeta-integrates-to-one}
\int_\pa \tilde \zeta_h(x;y) dy = 1.
\end{equation}

For the purpose of the analysis, it is convenient to decompose each iteration of Algorithm~\ref{alg:ideal} into four smaller conceptual pieces---smoothing step, smoothed sampling step, reweighting step, reweighted sampling step. 
At the beginning of the $(j+1)$\textsuperscript{th} iteration, the algorithm starts with the empirical distribution $\hat \xi^{[j]} \triangleq \frac{1}{N}\sum_{i=1}^{N}\delta_{X_i^{[j]}}$ of samples $X_1^{[j]}, \ldots, X_{N}^{[j]}$ from the previous iteration
.
In the first step (smoothing step) of the iteration, 
the algorithm produces a probability density $\xi^{[j+1/2]}$ by smoothing and regularizing the empirical measure $\hat \xi^{[j]}$%
.
In the second step (smoothed sampling step), 
the algorithm generates iid samples $X_1^{[j+1/2]}, \ldots, X_{N}^{[j+1/2]}$ from $\xi^{[j+1/2]}$ to obtain the empirical measure $\hat \xi^{[j+1/2]}\triangleq \frac{1}{N}\sum_{i=1}^{N} \delta_{X_i^{[j+1/2]}}$%
.
In the third step (reweighting step), 
the algorithm adjusts weights of the each probability mass of the empirical distribution to get a new distribution $\xi^{[j+1]} \triangleq \frac{1}{N}\sum_{i=1}^{N} w_i\delta_{X_i^{[j+1/2]}}$, which has the same support as $\hat \xi^{[j+1/2]}$ but shifted (via redistribution of the weights $w_i$'s) toward the target distribution%
. 
In the fourth step (reweighted sampling step),
the algorithm generates samples $X_1^{[j+1]}, \ldots, X_{N}^{[j+1]}$ from $\xi^{[j+1]}$ and constructs a new empirical distribution $\hat \xi^{[j+1]}$.
(Note that for design points $X_{i}^{[j]}$, we are using the superscript with square bracket to denote the index of the iteration and the subscript to denote the index of the sample within the iteration.)
The first two steps correspond to the perturbation  step in Algorithm~\ref{alg:ideal} and the last two steps correspond to the resampling step in Algorithm~\ref{alg:ideal}.
Schematically, the process can be summarized as follows:
\begin{align*}
0\text{\textsuperscript{th} iteration: }\qquad
&{\color{white}\hat \xi^{[0]} \ \ 
\stackrel{\text{smoothing}}{\longrightarrow} \ \, \, 
\xi^{[\ \,-1/2]} \ \ 
\stackrel{\text{sampling}}{\longrightarrow} }\ \ \; 
\hat \xi^{[\; \,-1/2]} \ \;
\stackrel{\text{reweighting}}{\longrightarrow} \ \
\xi^{[0]} \ \ 
\stackrel{\text{sampling}}{\longrightarrow} \ \ 
\hat \xi^{[0]} \ \ 
\\
1\text{\textsuperscript{st} iteration: }\qquad
&\hat \xi^{[0]} \ \ 
\stackrel{\text{smoothing}}{\longrightarrow} \ \ 
\xi^{[0+1/2]} \ \ 
\stackrel{\text{sampling}}{\longrightarrow} \ \ 
\hat \xi^{[0+1/2]} \ \ 
\stackrel{\text{reweighting}}{\longrightarrow} \ \
\xi^{[1]} \ \ 
\stackrel{\text{sampling}}{\longrightarrow} \ \ 
\hat \xi^{[1]} \ \ 
\\
2\text{\textsuperscript{nd} iteration: }\qquad
&\hat\xi^{[1]} \ \ 
\stackrel{\text{smoothing}}{\longrightarrow}\ \ 
\xi^{[1+1/2]} \ \ 
\stackrel{\text{sampling}}{\longrightarrow} \ \ 
\hat \xi^{[1+1/2]} \ \ 
\stackrel{\text{reweighting}}{\longrightarrow} \ \
\xi^{[2]} \ \ 
\stackrel{\text{sampling}}{\longrightarrow} \ \ 
\hat \xi^{[2]} \ \ 
\\
\vdots \qquad \qquad \ \ & 
\end{align*}
where $\hat\xi^{[\ -1/2]}$ is the empirical distribution of an arbitrary initial samples. A typical choice would be i.i.d.\ uniform samples from $\pa$. 
The precise description of the four steps is as follows: to generate samples from the target measure $\xi$ on the parameter space $\pa$, (or equivalently, $\eta$ on the manifold $\M$),
\begin{itemize}
	\item At iteration 0, we skip the smoothing and smoothed sampling step, and start directly with an arbitrary initial samples $X_1^{[\ -1/2]},\ldots,X_{N}^{[\ -1/2]}$. (Then we proceed to the reweighting step and the reweighted sampling step.)

	\item At iteration $j+1$, we start with an empirical distribution $\hat \xi^{[j]}$ of the samples $X_1^{[j]},\ldots,X_{N}^{[j]}$ from the previous iteration
	$$\hat\xi^{[j]} \triangleq \frac{1}{N}\sum_{i=1}^{N} \delta_{X_i^{[j]}}.$$
	Now, for suitably chosen parameters $h$ and $q$ (whose choice will be discussed later in this section and Section~\ref{sec:examples}),  
	\begin{itemize}\item[]\begin{itemize}
		\item[Step 1)] 
		Smooth out the empirical distribution $\hat\xi^{[j]}$ with $\zeta_h$ to get $\tilde\xi^{[j+1/2]}$
		$$
		\frac{d\tilde\xi^{[j+1/2]}}{d\lambda^m}(x) 
		\triangleq \frac{1}{N}\sum_{i=1}^{N}\tilde\zeta_h(x;X_i^{[j]}).
		$$
		Set $\xi^{[j+1/2]}$ as a mixture of the uniform distribution (on $\pa$) and $\tilde \xi^{[j+1/2]}$ with probability $q$ and $1-q$, respectively:
		$$
		\frac{d\xi^{[j+1/2]}}{d\lambda^m}(x) \triangleq q/\lambda^m(\pa) + (1-q)\frac{1}{N}\sum_{i=1}^{N}\tilde \zeta_h(x;X_i^{[j]}) 
		$$
		where $\lambda^m(\pa)$ denotes the $m$ dimensional volume of $\pa$.

		\item[Step 2)] Generate (for example, via acceptance-rejection) iid samples $X_1^{[j+1/2]},\ldots,X_{N}^{[j+1/2]}$ from $\xi^{[j+1/2]}$, and set $\hat \xi^{[j+1/2]}$ to be the empirical distribution of the generated samples:
		$$
		\hat\xi^{[j+1/2]} \triangleq \frac{1}{N}\sum_{i=1}^{N} \delta_{X_i^{[j+1/2]}}
		.
		$$

		\item[Step 3)] 
		Evaluate $f$ and $J_m f$ at $X_i^{[j]}$-s and re-distribute the weights as follows:
		\begin{equation}\label{eq:xi-j-plus-1}
		\xi^{[j+1]}
		\triangleq \sum_{i=1}^{N} \frac{(\mu\circ f\cdot J_m f \cdot \iota )(X_i^{[j+1/2]})}{\sum_{l=1}^{N} (\mu\circ f\cdot J_m f \cdot \iota)(X_l^{[j+1/2]})}\delta_{X_i^{[j+1/2]}},
		\end{equation}
		where $\iota \triangleq d\lambda^m/d\xi^{[j+1/2]}$ is the reciprocal of the density of $\xi^{[j+1/2]}$ w.r.t.\ the Lebesgue measure.
		Recall that $\mu$ is the Radon-Nikodym derivative of the target distribution $\eta$.

		\item[Step 4)] Generate  iid samples  $X_1^{[j+1]},\ldots,X_{N}^{[j+1]}$ from $\xi^{[j+1]}$ to get an empirical distribution 
		$$
		\hat\xi^{[j+1]} \triangleq \frac{1}{N} \sum_{i=1}^{N} \delta_{X_i^{[j+1]}}
		$$
	\end{itemize}\end{itemize}

	\item Repeat step 1)-4) to get $\hat\xi^{[j+2]}, \hat\xi^{[j+3]}, \ldots$
\end{itemize}
The main result of this section is that the above algorithm is consistent.
Let 
\begin{equation*}
\alpha(N) \triangleq 
\begin{cases}
N^{-1/2} & \mathrm{if}\ m = 1
\\
N^{-1/2}\log(N+1) & \mathrm{if}\ m = 2
\\
N^{-1/m} & \mathrm{otherwise}
\end{cases}
\end{equation*}
and, for any function $g:A \to \R^d$ on a domain $A$, let $\|g\|_\infty \triangleq \sup_{x\in A} |g(x)|$ where $|\cdot|$ is the Euclidean norm, and let $\|g\|_\lip \triangleq \sup_{x,y\in A,\,x\neq y} \frac{|g(x) - g(y)|}{|x-y|}$.
Let $D \triangleq \diam(\pa)$, $V \triangleq \lambda^m(\pa)$, and $c$ denote the constant from Proposition~\ref{prop:key_prop}  that depends only on $\|d\xi/d\lambda^m\|_\infty$, $\|d\xi/d\lambda^m\|_\lip$, $\|\nabla(d\xi/d\lambda^m)\|_\lip$, $q$, and $m$. 
\begin{theorem}\label{thm:main-result}
Suppose that $h$ is chosen in such a way that
$$c\alpha(N)\big(D^2V^2(\|\zeta_h\|_\lip + \|\nabla\zeta_h\|_\lip) + DV\|\zeta_h\|_\lip\big) < 1.$$
Algorithm~\ref{alg:ideal} is consistent in the sense that the Kantorovich-Rubinstein distance between the output $\hat\xi^{[j]}$ and the target measure $\xi$ converges to 0 in $L_1$ as $N\to\infty$ and $j\to\infty$.
More specifically,
\begin{align}
\E \W_1(\hat\xi^{[j]},\xi) 
&
\leq 
\frac{c\alpha(N)\big( (D^2V^2+ DV)h + (2D+D^2V)\big) }{1-c\alpha(N)\big(D^2V^2(\|\zeta_h\|_\lip + \|\nabla\zeta_h\|_\lip) + DV\|\zeta_h\|_\lip\big)}
\nonumber
\\
&
\qquad 
+\Big(c\alpha(N)\big(D^2V^2(\|\zeta_h\|_\lip
+ \|\nabla\zeta_h\|_\lip) + DV\|\zeta_h\|_\lip\big)\Big)^j\E \W_1(\hat \xi^{[0]},\xi).
\label{eq:main-result}
\end{align}
\end{theorem}
We defer the proof of Theorem~\ref{thm:main-result} to Section~\ref{sec:proofs}
and conclude this section with a few remarks regarding the implications of the theorem.

Note first that $\hat \xi^{[j]}$ will approach $\xi$ rapidly at the beginning and then linger around $\xi$ as $j$ increases. 
In view of this, one should choose the number of iterations $j$ in such a way that the two terms in \eqref{eq:main-result} are of the same order. 
That is, for a given $N$, $j$ should be chosen roughly at around
\begin{equation}\label{quantity:number-of-iterations}
\frac{\log\W_1(\hat \xi^{[0]}, \xi) - \log\Big(c\alpha(N)\big( (D^2V^2+ DV)h + (2D+D^2V)\big)\Big) }{-\log \Big(c\alpha(N)\big(D^2V^2(\|\zeta_h\|_\lip
+ \|\nabla\zeta_h\|_\lip) + DV\|\zeta_h\|_\lip\big)\Big)}.
\end{equation}
If $j$ is much smaller than \eqref{quantity:number-of-iterations}, the second term will dominate the error while it can be reduced geometrically and hence losing opportunities to reduce the total error efficiently; 
on the other hand, if $j$ is much larger than \eqref{quantity:number-of-iterations}, the first term will dominate the error and hence one will waste extra efforts without much gain in terms of the error. 
Note also that if the starting configuration $\hat \xi^{[0]}$ is reasonably close to $\xi$ to begin with, Algorithm~\ref{alg:ideal} will require very small number of iterations before it stabilizes. 
This is consistent with our numerical experiences reported in Section~\ref{sec:examples}.

There is a tradeoff between choosing $h$ small and large. 
Note that the limit supremum condition in Theorem~\ref{thm:main-result} requires that $h$ should not decrease to 0 too fast compared to the rate at which the size $N$ of samples grows.
On the other hand, in case $D$ and $V$ are large (for example, as in our Example~\ref{eg:exponential}, where $D=100\sqrt 2$ and $V = 10,000$), the error that does not decrease at a geometric rate w.r.t.\ the number of iterations---i.e., the first term in \eqref{quantity:number-of-iterations}---will be more or less proportional to $cD^2V^2\alpha(N)h$ for the practical range of values of $N$, and hence, small $h$ is desired.
That is, larger $h$ allows to make sure that the algorithm is stable w.r.t.\ the iterations and for smaller number of samples $N$; smaller $h$ allows one to reduce the final error of the algorithm after a sufficiently large number of iterations.

Finally, while our convergence bound is explicit, such explicitness seems to come at the expense of tight asymptotics.
Note that $\|\zeta_h\|_\lip$ and $\|\nabla \zeta_h\|_\lip$ are typically of order $h^{-m-1}$ and $h^{-m-2}$, respectively. 
That is, $h$ should not decrease to 0 at a faster rate than $N^{-\frac{1}{m(m+2)}}$ to satisfy the limit supremum condition. 
This is a very slow rate even for moderately large $m$'s, and we expect that the algorithm would be stable for $h$'s that decreases faster than such a rate guaranteed by Theorem~\ref{thm:main-result}. 
The practical choice of $h$ and $q$ will be further discussed in Section~\ref{sec:examples}.

%
%
\section{Proofs}\label{sec:proofs}
This section provides the proof of Theorem~\ref{thm:main-result}. 
For the notational convenience, we adopt the (common) convention that for a function $h:A\to\R$ and a measure $\mu$ on $A$,
$$
\mu h = \int_A h(y)\mu(dy).
$$
The proof of Theorem~\ref{thm:main-result} hinges critically on the following proposition.
\begin{proposition}\label{prop:key_prop}
%
There exists a constant $c$ depending only on $\|d\xi/d\lambda^m\|_\infty$, $\|d\xi/d\lambda^m\|_\lip$, $\|\nabla(d\xi/d\lambda^m)\|_\lip$, $q$, and $m$ such that
\begin{align*}
\E \W_1( \xi^{[j+1]}, \xi) 
&
\leq 
c\alpha(N)\big(D^2V^2(\|\zeta_h\|_\lip + \|\nabla\zeta_h\|_\lip) + DV\|\zeta_h\|_\lip\big) \E 
\W_1(\xi^{[j]}, \xi)
\\
&
\qquad
+ 
c\alpha(N)\big( (D^2V^2+ DV)h + (D+D^2V)\big) 
\end{align*}
where $D = \diam(\pa)$ and $V = \lambda^m(\pa)$.
\end{proposition}

Note that $\xi^{[j+1]}= \Psi_{G}(\hat \xi^{[j+1/2]})$ and $\xi = \Psi_{G}(\xi^{[j+1/2]})$, where $G =  (d\xi/d\lambda^m)/\allowbreak(d\xi^{[j+1/2]}/d\lambda^m)$.
%
%
%
%
In view of this, it is important to understand how smooth $\Psi_G(\nu)$ is w.r.t.\ $\nu$ in terms of $\W_1$ distance. 
Lemma~\ref{lem:new-mod_of_cont} provides a useful bound in this regard, and
it turns out that the bound involves quantities $\|d\xi / d\xi^{[j+1/2]}\|_\infty$ and $\|d\xi / d\xi^{[j+1/2]}\|_{\Lip}$. 
Lemma~\ref{lemma:new-l_infty_proximity} provides estimates for these quantities.
Before proving Proposition~\ref{prop:key_prop}, we first establish the following key lemmas.



\begin{lemma}\label{lem:new-mod_of_cont} 
\begin{equation}\label{eq:new-the_bound}
\begin{aligned}
\W_1(\xi^{[j+1]}, \xi)
\leq 
\big(\|d\xi/d\xi^{[j+1/2]}\|_\infty + 2\,\diam(\pa)\cdot\|d\xi/d\xi^{[j+1/2]}\|_{\Lip} \big)\, \W_1(\xi^{[j+1/2]}, \hat \xi^{[j+1/2]}).
\end{aligned}
\end{equation}
\end{lemma}

\begin{proof}
We first observe that from the straightforward bound $\|{fg}\|_{\Lip}\leq \|g\|_\infty \cdot \|f\|_{\Lip} + \|f\|_\infty \cdot \|g\|_{\Lip}$,
\begin{align}
&
\sup\{\nu(Gf) - \lambda(Gf): \|f\|_\Lip \leq 1, f(x^0) = 0\} 
\nonumber
\\
&
=
\sup_{f: \|f\|_\lip \leq 1, f(x^0) = 0}\|Gf\|_{\Lip}\left\{\nu\left(\frac{Gf}{\|Gf\|_{\Lip}}\right)-\lambda\left(\frac{Gf}{\|Gf\|_{\Lip}}\right)\right\}
\nonumber
\\
&
\leq 
\big(\|G\|_\infty + \diam(\pa) \cdot \|G\|_{\Lip}\big) \sup_{f: \|f\|_\lip \leq 1, f(x^0) = 0}\{\nu(f) - \lambda(f)\} 
\nonumber
\\
&
=
\big(\|G\|_\infty + \diam(\pa) \cdot \|G\|_{\Lip}\big)\, \W_1(\nu, \lambda)
.
\label{eq:new-slightly_different_from_W1}
\end{align}
Note also that for general measures $\nu$ and $\lambda$, a general potential $G$, and a general function $f$, 
\begin{align}
\Psi_G(\nu)f - \Psi_G(\lambda) f 
&= \frac{\nu(Gf)}{\nu(G)} - \frac{\lambda(Gf)}{\lambda(G)} 
\nonumber
\\
&= \frac{(\lambda(G)-\nu(G))\nu(Gf) + \nu(G)(\nu(Gf)- \lambda (Gf))}{\nu(G)\lambda(G)}
\nonumber
\\
&= \frac{(\lambda-\nu)(G)\Psi_G(\nu)f + (\nu-\lambda)(Gf)}{\lambda(G)}.
\label{eq:new-general_nu_lambda}
\end{align}
Now, fixing (arbitrarily chosen) $x^0 \in \pa$ and using (\ref{eq:new-slightly_different_from_W1}) and (\ref{eq:new-general_nu_lambda}), 
\begin{align*}
&
\lambda(G)\cdot \W_1(\Psi_G(\nu), \Psi_G(\lambda))  
\\
&
= \lambda(G)\cdot\sup \{\Psi_G(\nu)f - \Psi_G(\lambda)f:\|f\|_\lip \leq 1, f(x_0) = 0\} 
\\
&
= \sup\{(\nu-\lambda)(Gf) + (\lambda-\nu)(G)\Psi_G(\nu)f: \|f\|_\lip \leq 1, f(x^0) = 0\}
\\
&
\leq 
\sup\{\nu (Gf) - \lambda(Gf): \|f\|_\lip \leq 1, f(x^0) = 0\}
+ \big|{\nu}(G) - \lambda (G)\big|\,\sup\{\Psi_{G}(\nu) f: \|f\|_\lip \leq 1, f(x^0) = 0\}
\\
&
\leq 
\big(\|G\|_\infty + \diam(\pa)\cdot\|G\|_{\Lip} \big)\, \W_1(\nu, \lambda) 
+ \big|{\nu}(G) - \lambda (G)\big|\,\diam(\pa)
\\
&
\leq \big(\|G\|_\infty + 2\,\diam(\pa)\cdot\|G\|_{\Lip} \big)\, \W_1(\nu, \lambda) 
.
\end{align*}
We conclude that 
\begin{equation}\label{eq:new-mod_con_wrt_meas}
\W_1(\Psi_{G}(\nu), \Psi_{G}(\lambda))   
\leq \frac{1}{\lambda (G)}\big(\|G\|_\infty + 2\,\diam(\pa)\cdot\|G\|_{\Lip} \big)\,\W_1(\nu, \lambda) 
.
\end{equation}
If we set $G = d\xi/d\xi^{[j+1/2]}$, $\nu = \hat\xi^{[j+1/2]}$, and $\lambda = \xi^{[j+1/2]}$, then
$\Psi_G(\nu) = \xi^{[j+1]}$,
$\Psi_G(\lambda) = \xi$,
and
$\lambda(G) = 1$.
Therefore,
$$
\W_1(\xi^{[j+1]}, \xi)
\leq 
\big(\|d\xi/d\xi^{[j+1/2]}\|_\infty + 2\,\diam(\pa)\cdot\|d\xi/d\xi^{[j+1/2]}\|_{\Lip} \big)\, \W_1(\xi^{[j+1/2]}, \hat \xi^{[j+1/2]}).
$$
\end{proof}

\begin{lemma}\label{lemma:new-l_infty_proximity}
Let $\rho \triangleq d\xi/d\lambda^m$ denote the density of $\xi$ w.r.t.\  Lebesgue measure.
\begin{itemize}
\item[(i)] For each $x\in \pa$,
$$
\left|\frac{d\tilde\xi^{[j+1/2]}}{d\lambda^m}(x)-\frac{d\xi}{d\lambda^m}(x)\right| 
\leq 
\|\zeta_h\|_{\Lip}\, \W_1(\hat\xi^{[j]}, \xi) + h\,\|\rho\|_{\Lip}.
$$

\item[(ii)]
$$
\|d\xi / d\xi^{[j+1/2]}\|_\infty 
\leq \frac{1}{(1-q)\gamma} + \frac{\lambda^m(\pa)}{q(1-\gamma)}\Big(\|\zeta_h\|_{\Lip}\, \W_1(\hat\xi^{[j]},\xi) + h\,\|\rho\|_{\Lip}\Big)
$$
for each $\gamma \in (0,1)$.

\item[(iii)] For each $x\in \pa$,
$$
\left|\nabla\left(\frac{d\tilde\xi^{[j+1/2]}}{d\lambda^m}\right)(x)-\nabla\left(\frac{d\xi}{d\lambda^m}\right)(x)\right| 
\leq m\|{\nabla\zeta_h}\|_{\Lip}\, \W_1(\hat\xi^{[j]}, \xi) + mh\,\left\|\nabla\rho\right\|_{\Lip}.
$$

\item[(iv)]
\begin{align}
\big\|d\xi/d\xi^{[j+1/2]}\big\|_{\Lip}
&
\leq 
\frac{(1-q)}{(q/\lambda^m(\pa))^2}
\Big(\|\nabla \rho\|_\infty\cdot
\big\{\|\zeta_h\|_\lip \W_1(\hat\xi^{[j]},\xi)+ h\|\rho\|_\lip\big\}
\nonumber
\\
&
\qquad\qquad\qquad\quad
+
\|\rho\|_\infty\cdot
\big\{\|\nabla\zeta_h\|_\lip \W_1(\hat\xi^{[j]},\xi)+ h\|\nabla\rho\|_\lip\big\}
\Big)
+\frac{\lambda^m(\pa)}{q}\|\nabla \rho\|_\infty.
\nonumber
\end{align}

\end{itemize}
\end{lemma}
\begin{proof}
For (i), due to the construction of $\tilde \zeta_h$, we can use a similar argument as in Proposition 3.1 of \cite{bolley2007quantitative}. 
Note first that
$$
d\tilde\xi^{[j+1/2]}/d\lambda^m(x) 
= \frac{1}{N}\sum_{i=1}^N\tilde\zeta_h(x;X_i^{[j]}) = \int_{\pa}\tilde \zeta_h(x;y) \hat \xi^{[j]}(dy)$$
and hence
the difference between $d\tilde\xi^{[j+1/2]}/d\lambda^m(x)$ and $\int_\pa \tilde \zeta_h(x;y)\xi(dy)$ can be bounded as follows:
\begin{align}
\left|d\tilde\xi^{[j+1/2]}/d\lambda^m(x)-\int_\pa \tilde \zeta_h(x;y)\xi(dy)\right| 
&= 
\left|\int_\pa\tilde\zeta_h(x;y) (\hat\xi^{[j]}(dy) - \xi(dy)) \right|
\nonumber
\\
&
\leq \|\tilde\zeta_h(x;\cdot)\|_{\Lip}\,\W_1(\hat\xi^{[j]}, \xi).\label{ineq:d-xi-minus-zeta-conv-xi}
\end{align}
On the other hand, due to \eqref{zeta-integrates-to-one}, the distance between $\int_\pa \tilde \zeta_h(x;y)\xi(dy)$ and the density $d\xi/d\lambda^m$ of $\xi$ itself can be bounded in terms of the modulus of continuity of $d\xi/d\lambda^m$
\begin{align}
\left|\int_\pa \tilde \zeta_h(x;y)\xi(dy) - d\xi/d\lambda^m(x)\right| 
&= \left|\int_\pa \tilde\zeta_h (x;y) d\xi/d\lambda^m(y)dy-\int_\pa \tilde \zeta_h(x;y) d\xi/d\lambda^m(x)dy\right| 
\nonumber
\\
&
\leq 
\int_\pa \tilde\zeta_h(x;y)\,\big|d\xi/d\lambda^m(y) - d\xi/d\lambda^m(x)\big|dy
\nonumber
\\
&
\leq 
\sup_{\substack{y\in \pa\\|x-y|\leq h}}|d\xi/d\lambda^m(x)-d\xi/d\lambda^m(y)| \int_\pa \tilde\zeta_h(x;y)dy
\nonumber
\\
&\leq h \,\|d\xi/d\lambda^m\|_{\Lip}.
\label{ineq:zeta-conv-xi-minus-d-xi}
\end{align}
Now by triangle inequality and \eqref{ineq:d-xi-minus-zeta-conv-xi} and \eqref{ineq:zeta-conv-xi-minus-d-xi}, we arrive at the conclusion of (i). 

Turning to (ii),
\begin{align*}
\|d\xi / d\xi^{[j+1/2]}\|_\infty 
&
=
\left\| \frac{d\xi/d\lambda^m}{d\xi^{[j+1/2]}/d\lambda^m}\right\|_\infty
=
\sup_{x\in \pa} \frac{\rho(x)}{q/\lambda^m(\pa) + (1-q)\frac{d\tilde \xi^{[j+1/2]}}{d\lambda^m}(x)  }
\\
&
=
\sup_{x\in \pa} \frac{\rho(x)}{q/\lambda^m(\pa) + (1-q)\frac{d\tilde \xi^{[j+1/2]}}{d\lambda^m}(x)  }
\left[\one{d\tilde\xi/d\lambda^m (x) > \gamma \rho(x)}+\one{d\tilde\xi/d\lambda^m (x) \leq \gamma \rho(x)} \right]
\\
&
\leq
\sup_{x\in \pa}\left[ \frac{\rho(x)}{q/\lambda^m(\pa) + (1-q)\gamma\rho(x)}\one{d\tilde\xi/d\lambda^m (x) > \gamma \rho(x)}
+\frac{\rho(x)}{q/\lambda^m(\pa)}\one{d\tilde\xi/d\lambda^m (x) \leq \gamma \rho(x)} \right]
\\
&
=
\frac{1}{ (1-q)\gamma} + \sup_{x\in \pa}  \frac{\rho(x)}{q/\lambda^m(\pa)}\one{\rho(x) - \|\zeta_h\|_{\Lip} \W_1(\hat\xi^{[j]}, \xi)- h\,\|\rho\|_{\Lip} \leq \gamma \rho(x)} 
\\
&
=
\frac{1}{ (1-q)\gamma} + \sup_{x\in \pa}  \frac{\rho(x)}{q/\lambda^m(\pa)}\one{\rho(x)\leq \frac{ \|\zeta_h\|_{\Lip} \W_1(\hat\xi^{[j]}, \xi)+ h\,\|\rho\|_{\Lip} }{1-\gamma}} 
\\
&
\leq 
\frac{1}{ (1-q)\gamma} + \frac{\|\zeta_h\|_{\Lip} \W_1(\hat\xi^{[j]}, \xi)+ h\,\|\rho\|_{\Lip} }{q(1-\gamma)/\lambda^m(\pa)}.
\end{align*}

For (iii),
note that 
\begin{align*}
\left|
\partial_i(d\tilde\xi^{[j+1/2]}/d\lambda^m)(x) 
-
\int_\pa \frac{\partial}{\partial x_i} \tilde \zeta_h(x;y)\, \xi(dy)
\right|
&=
\left|
\int_\pa \frac{\partial}{\partial x_i} \tilde \zeta_h (x;y)\, \hat \xi^{[j]}(dy) 
- \int_\pa \frac{\partial}{\partial x_i} \tilde \zeta_h(x;y)\, \xi(dy)
\right|
\\
&
\leq
\left\|\frac{\partial}{\partial x_i} \tilde \zeta_h(x;\cdot)\right\|_{\Lip}\,\W_1(\hat \xi^{[j]}, \xi)
\\
&\leq 
\left\|\partial_i \zeta_h\right\|_{\Lip}\,\W_1(\hat \xi^{[j]}, \xi)
.
\end{align*}
Let $\tilde \pa \triangleq \prod_{i=1}^m \big[x_{\min}^i-h,\ x_{\max}^i+h\big]$ be the set obtained by fattening $\pa$ by $h$ along each coordinate and
$\tilde \rho: \tilde \pa\to \R_+$ be the extension of $\rho$ from $\pa$ to $\tilde \pa$ by reflection; more specifically,
$$
\tilde \rho(x) \triangleq \rho(\tilde x) = \rho(\tilde x^1,\ldots,\tilde x^m)
\qquad \text{where} \qquad
\tilde x^j
= 
\begin{cases}
2x_{\min}^j-x^j
&
\text{if }\quad \phantom{0\leq}\hfill x^j < x_{\min}^j;
\\
\hfill x^j
&
\text{if }\quad x^j_{\min} \leq x^j < x^j_{\max};
\\
2x_{\max}^j-x^j
&
\text{if }\quad x_{\max}^j \leq x^j\hfill.
\end{cases}
$$
Note that \eqref{zeta-is-zero} implies $\zeta_h(x-y)= 0$ for $x\in \pa$ and $y\in\partial \tilde \pa$.
From this along with the integration by parts formula,
\begin{align*}
\left|
\int_\pa \frac{\partial}{\partial x_i} \tilde \zeta_h(x;y)\, \xi(dy)
- 
\int_{\tilde \pa}  \zeta_h(x-y)\, \partial_i\tilde\rho (y)dy
\right|
&
=
\left|
\int_{\tilde \pa} \partial_i \zeta_h(x-y)\, \tilde \rho(y) dy
- 
\int_{\tilde \pa}  \zeta_h(x-y)\, \partial_i\tilde \rho (y)dy
\right|
\\
&
=
\left|
\int_{\partial \tilde \pa} \zeta_h(x-y) \tilde \rho(y) d\partial \tilde \pa
\right|
=0.
\end{align*}
The integration by parts formula holds near the boundary $\partial \pa$ of $\pa$ as well since $\tilde\rho$ is continuously differentiable at the boundary due to the assumption A2.  
Finally,
\begin{align*}
\left|
\int_{\tilde \pa}  \zeta_h(x-y)\, \partial_i\tilde \rho (y)dy
-
\partial_i\tilde \rho (x)
\right|
&
=
\left|
\int_{\tilde \pa}  \zeta_h(x-y)\, \partial_i\tilde \rho (y)dy
-
\int_{\tilde \pa}  \zeta_h(x-y)\, \partial_i\tilde \rho (x)dy
\right|
\\
&
\leq
h\|\partial_i \tilde \rho\|_{\Lip} = h\|\partial_i \rho\|_{\Lip}.
\end{align*}
Combining the above inequalities, we arrive at (iii).

%
Turning to (iv), note that in general for any smooth $f, g$ and positive constants $q, v$, 
\begin{align*}
\Big\|\frac{f}{(1-q)g+q/v}\Big\|_\lip 
&
= \Big\|\nabla\Big(\frac{f}{(1-q)g+q/v}\Big)\Big\|_\infty 
= \Big\|\frac{(\nabla f)((1-q)g + q/v) - (1-q)f(\nabla g)}{((1-q)g + q/v)^2}
\Big\|_\infty
\\
&
\leq
\Big\|
(1-q)\frac{(\nabla f)g  - f (\nabla g)}{(q/v)^2}
+ \frac{(\nabla f) (q/v)}{((1-q)g + q/v)^2}
\Big\|_\infty
\nonumber
\\
&
\leq
\Big\|
(1-q)\frac{(\nabla f)(g-f)  + f (\nabla f-\nabla g)}{(q/v)^2}\Big\|_\infty
+  \Big\|(v/q)\nabla f \Big\|_\infty.
\end{align*}
Substituting $f$, $g$, $v$ with $d\xi/d\lambda^m$, $d\xi^{[j+1/2]}/d\lambda^m$, $\lambda^m(\pa)$,
\begin{align}
&
\big\|d\xi/d\xi^{[j+1/2]}\big\|_{\Lip}
\nonumber
\\
&
\leq 
\frac{(1-q)}{(q/\lambda^m(\pa))^2}
\left(\| \rho\|_\lip\cdot \Big\|\frac{d\xi}{d\lambda^m} - \frac{d\xi^{[j+1/2]}}{d\lambda^m}\Big\|_\infty 
+
\|\rho\|_\infty\cdot
\Big\|\nabla\Big(\frac{d\xi}{d\lambda^m}\Big)-\nabla\Big(\frac{d\xi^{[j+1/2]}}{d\lambda^m}\Big)\Big\|_\infty
\right)
\nonumber
\\
&
\qquad\qquad
+(\lambda^m(\pa)/q)\Big\|\nabla \Big(\frac{d\xi}{d\lambda^m}\Big)\Big\|_\infty
\nonumber
\\
&
\leq 
\frac{(1-q)}{(q/\lambda^m(\pa))^2}
\left(\|\rho\|_\lip\cdot
\{\|\zeta_h\|_\lip \W_1(\hat\xi^{[j]},\xi)+ h\|\rho\|_\lip\}
+\|\rho\|_\infty\cdot
\{\|\nabla\zeta_h\|_\lip \W_1(\hat\xi^{[j]},\xi)+ h\|\nabla\rho\|_\lip\}
\right)
\nonumber\\
&
\qquad\qquad
+(\lambda^m(\pa)/q)\|\nabla \rho\|_\infty.
\nonumber
\end{align}
This concludes the proof.
\end{proof}

With Lemma~\ref{lem:new-mod_of_cont} and Lemma~\ref{lemma:new-l_infty_proximity} in hand, the proof of the Proposition~\ref{prop:key_prop} becomes straightforward.
\begin{proof}[Proof of Proposition~\ref{prop:key_prop}]
Note that (ii) and (iv) of Lemma~\ref{lemma:new-l_infty_proximity} imply that there exists a constant $c_1$ not depending on $\lambda^m(A)$ and $h$ such that
$$
\|d\xi / d\xi^{[j+1/2]}\|_\infty
\leq c_1 \lambda^m(A)\cdot \|\zeta_h\|_\lip\cdot \W_1(\hat\xi^{[j]},\xi) + c_1(\lambda^m(\pa))^2 + c_1 h \lambda^m(\pa) + c_1
$$
and
$$
\|d\xi / d\xi^{[j+1/2]}\|_\lip 
\leq c_1 (\lambda^m(\pa))^2 (\|\zeta_h\|_\lip  + \|\nabla \zeta_h\|_\lip ) \W_1(\hat\xi^{[j]},\xi) + c_1(\lambda^m(\pa))^2 + c_1 \lambda^m(\pa).
$$
Therefore, Lemma~\ref{lem:new-mod_of_cont} implies that
\begin{align*}
\W_1(\xi^{[j+1]}, \xi)
&
\leq 
c_1\big(DV^2(\|\zeta_h\|_\lip + \|\nabla\zeta_h\|_\lip) + V\|\zeta_h\|_\lip\big) \W_1(\xi^{[j]}, \xi)
\, \W_1(\xi^{[j+1/2]}, \hat \xi^{[j+1/2]}) 
\\
&
\qquad
+ 
c_1\big( (DV^2+ V)h + (1+DV)\big)\, \W_1(\xi^{[j+1/2]}, \hat \xi^{[j+1/2]}) 
\end{align*}
where $c_1$ is a constant that depends only on $q$ and $\xi$, $D$ denotes $\diam(\pa)$, and $V$ denotes $\lambda^m(\pa)$.
From Theorem 1 of \cite{fournier2015rate}, we see that $\E\Big[\W_1(\xi^{[j+1/2]}, \hat \xi^{[j+1/2]})\Big| \xi^{[j]}\Big]$ can be bounded by $c_2\, \diam(\pa)\alpha(N)$ where $c_2$ is a constant depending only on $m$.
Therefore,
\begin{align*}
\E\W_1(\xi^{[j+1]}, \xi)
&
\leq 
\E \Big[\E\Big[c_1\big(DV^2(\|\zeta_h\|_\lip + \|\nabla\zeta_h\|_\lip) + V\|\zeta_h\|_\lip\big) \W_1(\hat\xi^{[j]}, \xi)
\, \W_1(\xi^{[j+1/2]}, \hat \xi^{[j+1/2]}) 
\Big|
\xi^{[j]}
\Big]\Big]
\\
&
\qquad
+ 
\E \Big[\E\Big[
c_1\big( (DV^2+ V)h + (1+DV)\big)\, \W_1(\xi^{[j+1/2]}, \hat \xi^{[j+1/2]}) 
\Big|
\xi^{[j]}
\Big]\Big]
\\
&
\leq 
c_1\E \Big[\E\Big[\W_1(\xi^{[j+1/2]}, \hat \xi^{[j+1/2]}) 
\Big|
\xi^{[j]}
\Big]\big(DV^2(\|\zeta_h\|_\lip + \|\nabla\zeta_h\|_\lip) + V\|\zeta_h\|_\lip\big) \W_1(\hat\xi^{[j]}, \xi)
\Big]
\\
&
\qquad
+ 
c_1\E \Big[\E\Big[\W_1(\xi^{[j+1/2]}, \hat \xi^{[j+1/2]}) 
\Big|
\xi^{[j]}
\Big]
\big( (DV^2+ V)h + (1+DV)\big) \Big]
\\
&
\leq 
c_1c_2\alpha(N)\big(D^2V^2(\|\zeta_h\|_\lip + \|\nabla\zeta_h\|_\lip) + DV\|\zeta_h\|_\lip\big) \E 
\W_1(\hat\xi^{[j]}, \xi)
\\
&
\qquad
+ 
c_1c_2\alpha(N)\big( (D^2V^2+ DV)h + (D+D^2V)\big) 
\end{align*}
Therefore, the conclusion of the proposition follows. 
\end{proof}

Now we are ready to prove Theorem~\ref{thm:main-result}.

\begin{proof}[Proof of Theorem~\ref{thm:main-result}]
Again, from Proposition~\ref{prop:key_prop} and Theorem~1 of \cite{fournier2015rate}, $\E
	\W_1(\hat\xi^{[j+1]}, \xi^{[j+1]})\leq c D \alpha(N)$, and hence,
\begin{align*}
\E \W_1(\hat \xi^{[j+1]}, \xi) 
&
\leq \E\left[
	\W_1(\hat\xi^{[j+1]}, \xi^{[j+1]})
	+ \W_1(\xi^{[j+1]}, \xi)
	\right]
\\
&
\leq 
c\alpha(N)\big( (D^2V^2+ DV)h + (2D+D^2V)\big) 
\\
&
\qquad
+
c\alpha(N)\big(D^2V^2(\|\zeta_h\|_\lip + \|\nabla\zeta_h\|_\lip) + DV\|\zeta_h\|_\lip\big) \E 
\W_1(\hat\xi^{[j]}, \xi)
\end{align*}
for some $c$. 
Therefore, $w_j \triangleq \E \W_1(\hat \xi^{[j]}, \xi)$ satisfies the following recursive inequality:
$$
w_{j+1} \leq a + b w_j
$$
where $a = c\alpha(N)\big( (D^2V^2+ DV)h + (2D+D^2V)\big) $ and $b = c\alpha(N)\big(D^2V^2(\|\zeta_h\|_\lip + \|\nabla\zeta_h\|_\lip) + DV\|\zeta_h\|_\lip\big)$. 
Solving this recursion, we get $w_{j} \leq \frac{a}{1-b} + b^j w_0$. 
That is,
\begin{align*}
\E \W_1(\hat\xi^{[j]},\xi) 
&
\leq 
\frac{c\alpha(N)\big( (D^2V^2+ DV)h + (2D+D^2V)\big) }{1-c\alpha(N)\big(D^2V^2(\|\zeta_h\|_\lip + \|\nabla\zeta_h\|_\lip) + DV\|\zeta_h\|_\lip\big)}
\\
&
\qquad 
+\Big(c\alpha(N)\big(D^2V^2(\|\zeta_h\|_\lip
+ \|\nabla\zeta_h\|_\lip) + DV\|\zeta_h\|_\lip\big)\Big)^j\E \W_1(\hat \xi^{[0]},\xi).
\end{align*}
\end{proof}

%
%
\section{Examples}\label{sec:examples}

In this section, we briefly discuss the choice of algorithmic parameters $h$, $q$, and $b$, and examine the numerical behavior of the algorithm with a few examples. 
Due to the conservative nature of our convergence analysis in Theorem~\ref{thm:main-result} and Lemma~\ref{lemma:key_lemma}, we do not provide definitive rules for the choice of the above parameters. 
Instead, we provide heuristic discussions and rough guidelines from our numerical experience here.
More thorough investigation will be pursued in subsequent studies.
As pointed out in Section~\ref{sec:consistency-convergence-analysis}, the choice of $h$ is critical for the stability and the performance of the algorithm. 
Although our sufficient condition in Theorem~\ref{thm:main-result} is conservative, our numerical experience confirms that a liberal choice of $h$ can indeed lead to unstable behavior of the algorithm.
That is, if $h$ is chosen too small compared to $N$, the algorithm may diverge from the target distribution.  
An indication of such divergence is clustering of the points, which can easily be detected with various methods. 
In view of this, we suggest starting with a conservative choice of $h$ and gradually reducing the size of $h$ as the algorithm stabilizes. 
When the clustering behavior is detected, the experimenter can either increase $N$ or increase $h$ so that the algorithm does not exhibit clustering behavior.
Turning to $q$, note that $q$ being away from 0 prevents $\xi^{[j+1/2]}$ from being much smaller than $\xi$ so that the ratio doesn't blow up. 
On the other hand, if $q$ is close to 1, our algorithm is not assigning enough resource (samples) in learning the geometry of the manifold. 
Such a trade-off can be noticed in the upper bounds in (ii) and (iv) of Lemma~\ref{lemma:new-l_infty_proximity} where both $q$ and $1-q$ appear in the denominator.
In view of this, we suggest choosing $q$ inside of the interior of $(0,1)$ sufficiently away from the boundary, say, between $1/10$ and $1/2$.
The choice of $b$ does not seem to make much difference in terms of the performance of the algorithm as far as $b$ is chosen sufficiently large. 


\begin{figure}[tb]
\hspace{-55pt}\includegraphics[width = 1.2\linewidth]{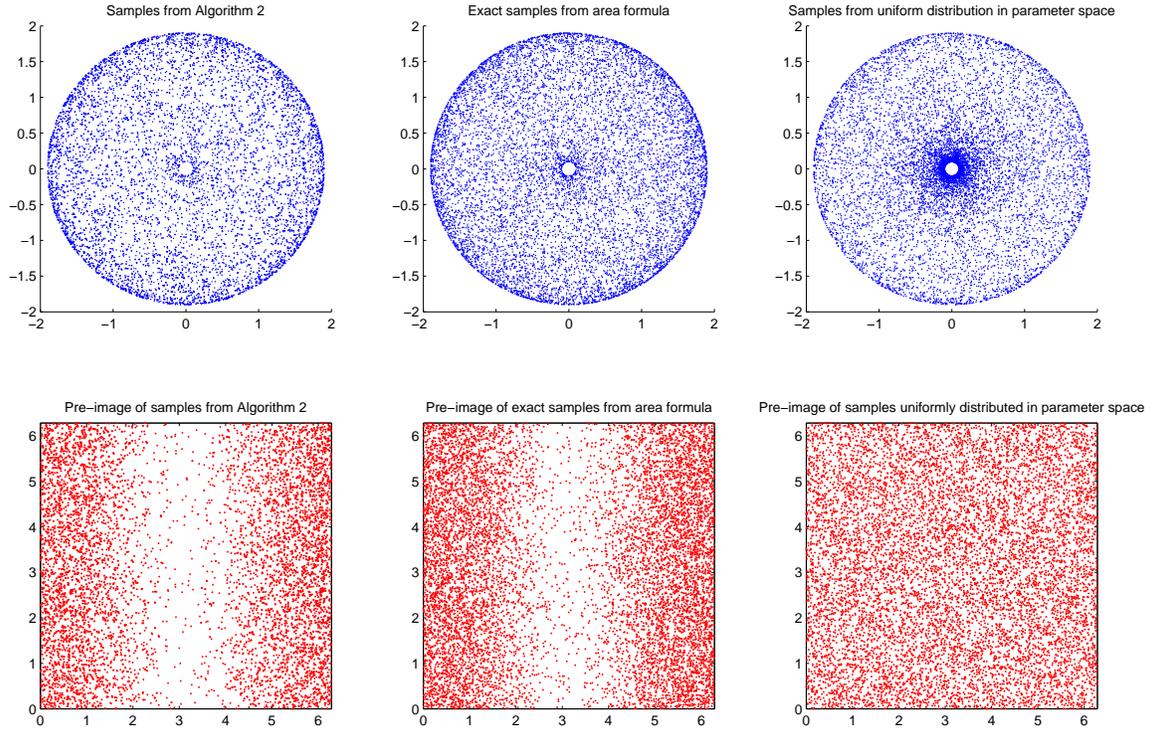}
\caption{(Uniform Samples From Torus) Comparison of the 10,000 samples from Algorithm~\ref{alg:ideal}, exact uniform distribution on the manifold, and uniform distribution in parameter space.\label{fig:torus}}
\end{figure}

\begin{example}{(Uniform Samples from Torus)}\label{eg:torus}
\cite{Diaconis2013} illustrate how to sample from a torus using the area formula (\ref{eq:area_formula}). Consider a torus
$$\mathcal M = \{((R+r\cos \theta) \cos\psi, (R+r\cos \theta)\sin\psi, r \sin \theta): 0\leq\theta,\psi<2\pi \}$$
where $0< r < R$. The major radius $R$ is the distance from the center of the tube to the center of the torus, and the minor radius $r$ is the radius of the tube. 
One way to parametrize $\mathcal M$ and its $2$-dimensional Jacobian $J_2f$ are
$$ f(\theta,\psi) = ((R+r\cos \theta) \cos\psi, (R+r\cos \theta)\sin\psi, r \sin \theta),\qquad J_2 f(\theta,\psi) =  r(R + r\cos\theta).$$
In view of (\ref{eq:area_formula}), one can generate exactly uniform samples on $\mathcal M$ w.r.t.\ Hausdorff measure by generating samples on $[0,2\pi]\times[0,2\pi]$ from the density $g(\theta,\psi) \propto R+r\cos\theta$.
For example, one can generate $\psi$ from the uniform distribution on $[0,2\pi]$, and (independently) generate $\theta$ from the density $\frac{1}{2\pi R}(R+r\cos\theta)$, via acceptance-rejection or inversion.
We compare the three different ways of covering $\mathcal M$ for the purpose of illustration of the consistency of our algorithm. 
Figure~\ref{fig:torus} compares the samples on $\mathcal M$ produced for $R=1$ and $r=0.9$.
The upper plot shows the samples projected on $x$-$y$ plane, and the lower plot shows the pre-image of the samples in the parameter space.
The plots on the left show the samples generated from Algorithm~\ref{alg:ideal} after the resampling step in the second iteration with $q = 0.1$ and $h = 0.5$.
The plots in the middle were produced with the 10,000 samples generated by the area formula (as described above), and the right plots show 10,000 samples generated by uniformly sampling in the parameter space, i.e., $\theta \sim U[0,2\pi]$ and $\psi \sim U[0,2\pi]$.
Observe that the left and middle plots are similar to each other while in the upper right plot the center of the torus is much more densely populated compared to the outer part of the torus.
This illustrates that the samples generated uniformly from the parameter space $A$ is far from uniform on the manifold $\mathcal M$, and Algorithm~\ref{alg:ideal} produces samples from the target distribution.
Figure~\ref{fig:theta_torus} compares the histogram of $\theta$ sampled by Algorithm~\ref{alg:ideal}, and the exact marginal of the target density $g(\theta,\psi) = \frac{1}{4\pi^2 R}(R+r\cos\theta)$.

\begin{figure}[tb]
\centering
\includegraphics[width = 0.8\linewidth]{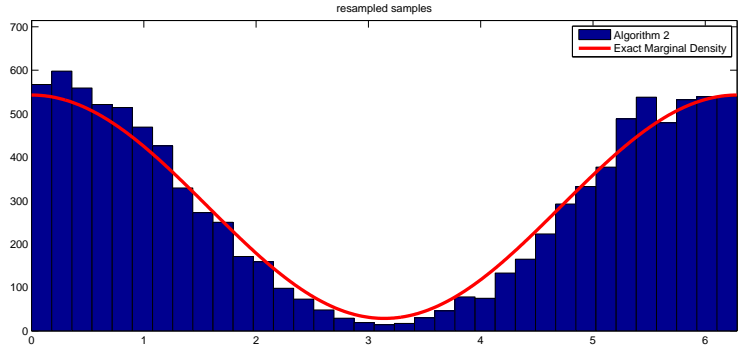}
\caption{(Uniform Samples From Torus) Histogram from 10,000 samples of $\theta$'s generated by Algorithm~\ref{alg:ideal}. Red line shows the exact target marginal density computed from the area formula.\label{fig:theta_torus}}
\end{figure}

\end{example}

\begin{example}{(Non-uniform Density on Torus)}\label{eg:torus_non_unif} In this example, we consider the same manifold $\mathcal M$ as in Example~\ref{eg:torus}, but we illustrate how Algorithm~\ref{alg:ideal} performs for a non-uniform density on $\mathcal M$.
Suppose that we are particularly interested in studying $\mathcal M$ in the proximity of a given point. 
For example, suppose that we are interested in $(0,1,0)$, and hence, we want to use more computational resource for the closer parts of the manifold to the point, and less resource for the farther parts of the manifold.
For this purpose, we choose a density proportional to the reciprocal of the squared distance from $(0,1,0)$.
More specifically, we want to sample from the distribution $P(d\mathbf x) = r(\mathbf x) \mathcal H^2(d\mathbf x)$ where $r(x,y,z) \propto 1/(x^2+(y-1)^2 + z^2)$.
Again, for this simple example, one can generate exact samples from $P$ directly from the area formula via acceptance-rejection with the proposal density proportional to
$$r(f(\theta,\psi))g(\theta,\psi)
\propto
\frac{R+r\cos\theta}{((R+r\cos\theta)\cos\psi)^2 + ((R+r\cos\theta)\sin\psi-1)^2 + (r\sin\theta)^2}.$$
The samples produced by Algorithm~\ref{alg:ideal} (after the third resampling step with $q=0.1$ and $h=0.5$), the exact samples generated by the area formula, and the samples generated uniformly in the parameter space are compared in Figure~\ref{fig:non_uniform_torus} and \ref{fig:non_uniform_torus_hist}.
Once can again it should be noted that Algorithm~\ref{alg:ideal} generates the correct distribution.
\begin{figure}[htb]
\hspace{-55pt}\includegraphics[width = 1.2\linewidth, scale=0.8]{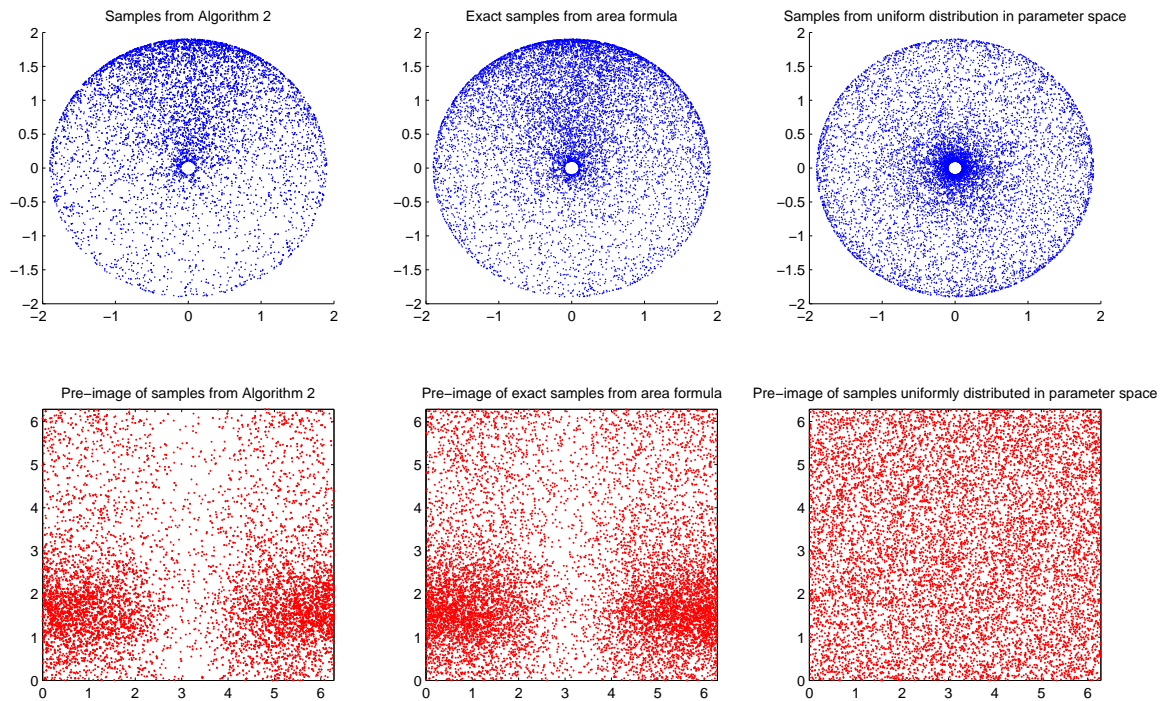}
\caption{(Non-Uniform Density on Torus) Comparison of the 10,000 samples from Algorithm~\ref{alg:ideal}, area formula, and uniform distribution in parameter space.\label{fig:non_uniform_torus}}
\end{figure}

\begin{figure}[htb]
\hspace{-35pt}
\includegraphics[width = 1.1\linewidth, height=7cm]{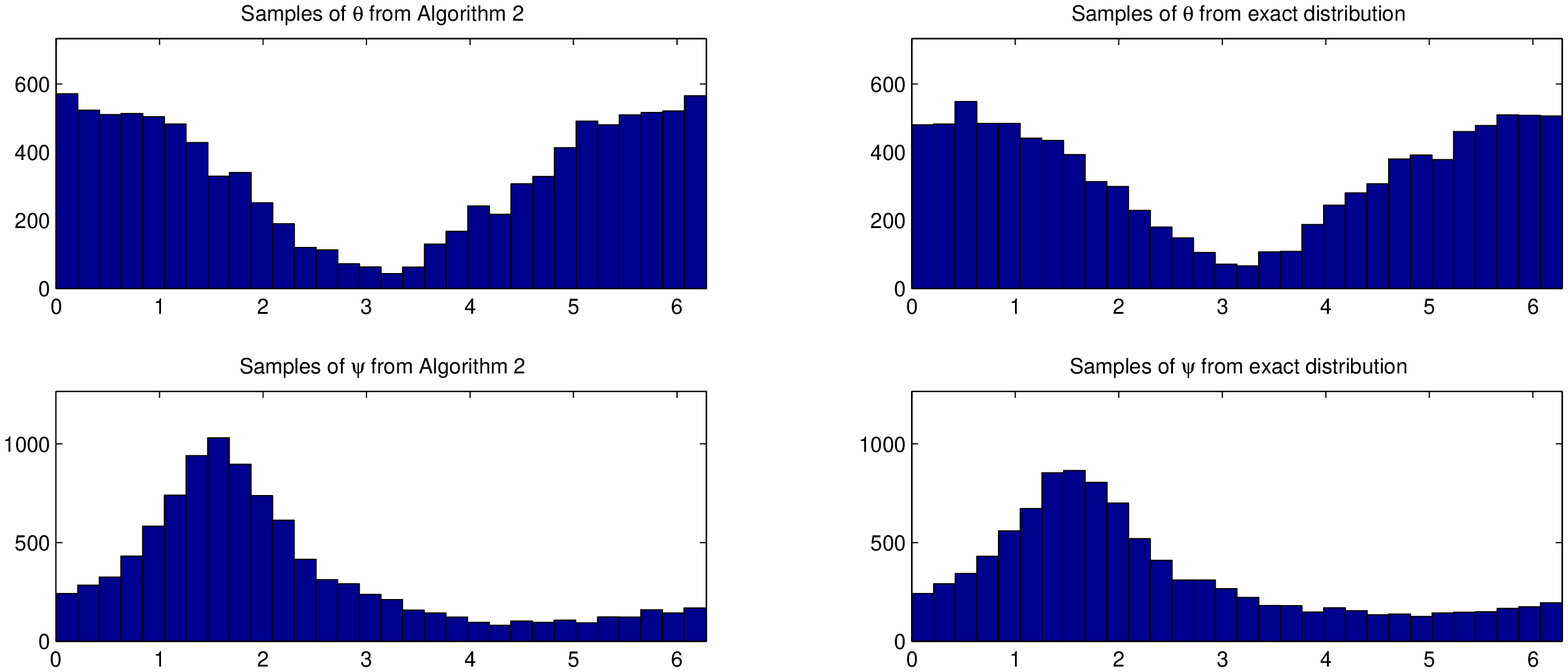}
\caption{(Non-Uniform Density on Torus) Histograms from 10,000 samples generated by Algorithm~\ref{alg:ideal}, and the exact area formula.\label{fig:non_uniform_torus_hist}}
\end{figure}
\end{example}

Next, we examine a more interesting case, 
where the model changes its behavior significantly on a small part of the parameter space while it remains relatively constant over the majority of the parameter space.
That is, most of the parameter space is mapped to a small fraction of the manifold and the rest---the majority---of the manifold comes from a small fraction of the parameter space. 
The next example illustrates how our algorithm discovers such  a small region of the parameter space.

\begin{example}\label{eg:exponential} (Exponential Model)
Here we consider a manifold
$$\mathcal M = \{(e^{-\theta t_1} + e^{-\psi t_1},e^{-\theta t_2} + e^{-\psi t_2}, e^{-\theta t_3} + e^{-\psi t_3}): \theta, \psi \in [0,100]\}$$
where $0 < t_1 < t_2 < t_3$. The first derivative
\begin{equation}
Df(\theta, \psi) = \left(
\begin{array}{cc}
-t_1 \exp(-\theta t_1) & -t_1 \exp(-\psi t_1)\\
-t_2 \exp(-\theta t_2) & -t_2 \exp(-\psi t_2)\\
-t_3 \exp(-\theta t_3) & -t_3 \exp(-\psi t_3)
\end{array}
\right),
\end{equation}
and the 2-dimensional Jacobian
\begin{equation}
J_2f(\theta,\psi)
=
\sqrt{\alpha_{12} + \alpha_{13} + \alpha_{23}}
\end{equation}
where
\begin{equation}
\alpha_{ij} =
t_i^2 t_j^2 \exp\{-2(\theta t_j + \psi t_i)\}\{\exp(\theta-\psi)(t_j-t_i) -1\}^2.
\end{equation}
Figure~\ref{fig:exponential} shows the result for $t_1 = 1$, $t_2 = 2$, $t_3 = 4$ with $2,000$ samples.
The left plot was produced by Algorithm~\ref{alg:ideal} with $q = 0.1$ and $h = 1$, (the upper plot shows the samples projected on a plane perpendicular to the vector $(0.5,-1,0.5)$, and the lower plot shows the pre-image of the samples in the parameter space); 
the plots in the middle show the 2,000 samples generated by the area formula; the right plots show 2,000 samples generated by uniformly sampling in the parameter space, i.e., $\theta, \psi \sim \text{Unif\,}\big([0,100]\times[0,100]\big)$.
The samples generated uniformly in the parameter space are concentrated in a small part of the boundary of the manifold.
The Algorithm~\ref{alg:ideal} was started with initial samples distributed uniformly in the parameter space, and the final samples were obtained after the resampling step in the 10th iteration. The progression of the algorithm is illustrated in Figure~\ref{fig:progression}.

\begin{figure}[htb]
\hspace{-60pt}\includegraphics[width = 1.2\linewidth, height = 11cm]{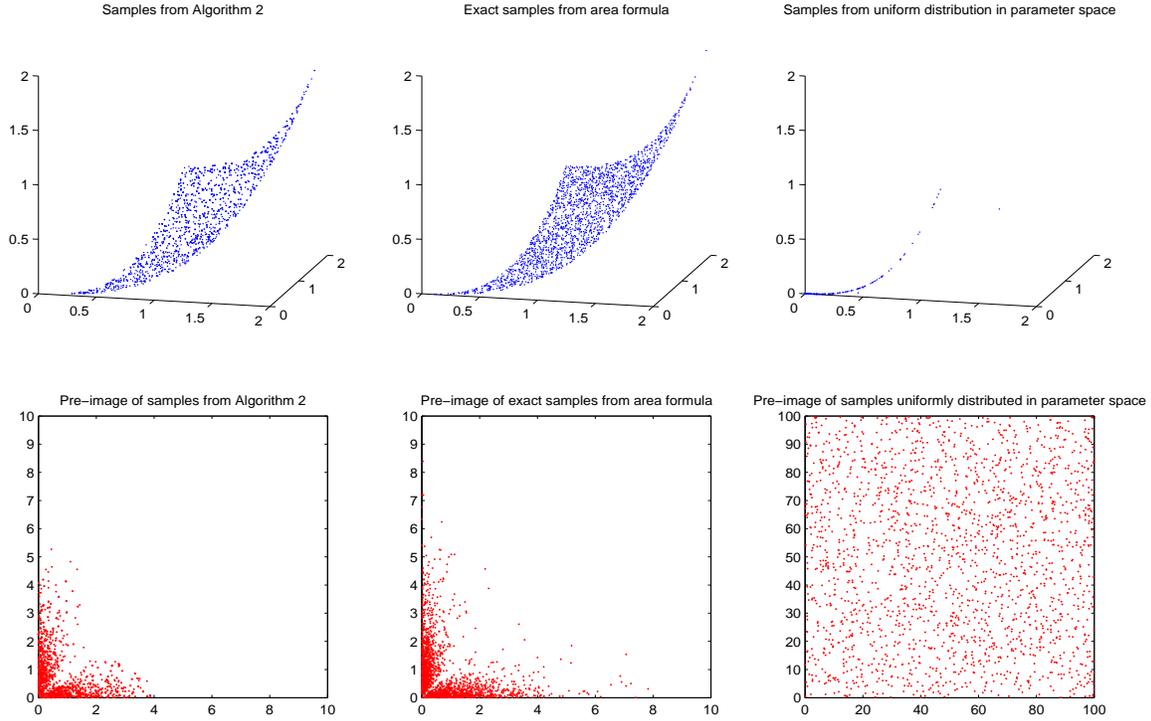}
\caption{(Exponential Model) Comparison of the 2,000 samples from Algorithm~\ref{alg:ideal}, area formula, and uniform distribution in parameter space.\label{fig:exponential}}
\end{figure}

\begin{figure}[htb]
\hspace{-35pt}\includegraphics[width = 1.1\linewidth]{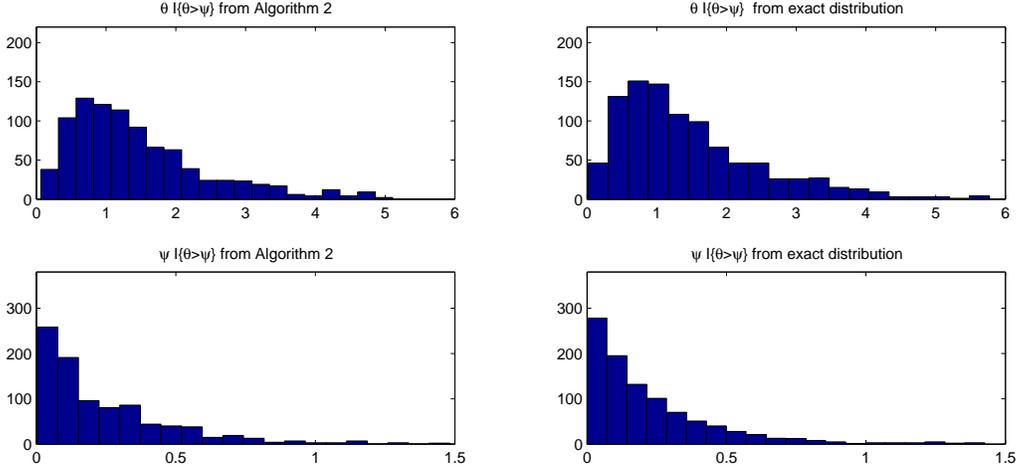}
\caption{(Exponential Model) Histograms from 2,000 samples generated by Algorithm~\ref{alg:ideal} and the exact area formula.}\label{fig:exp_hist}
\end{figure}

\begin{figure}[htb]
\hspace{-55pt}\includegraphics[width = 1.2\linewidth, height = 9cm]{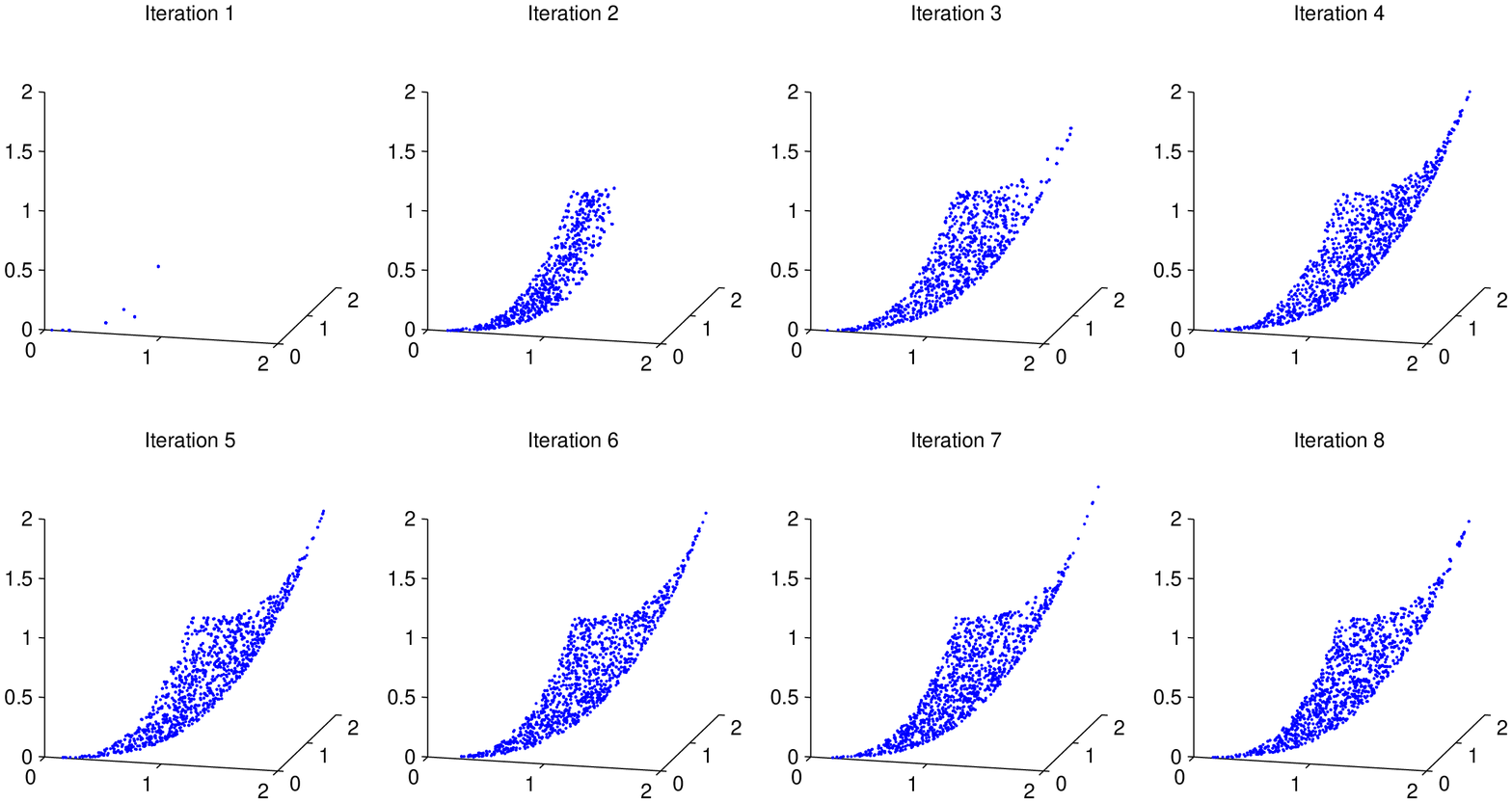}
\caption{(Exponential Model) The progression of Algorithm~\ref{alg:ideal}.\label{fig:progression}}
\end{figure}
\end{example}

\begin{example}\label{eg:sys-bio}{(ODE Models in Systems Biology)}
The dynamics of the enzymatic regulatory systems are often modeled with a set of ordinary differential equations. One of the most popular form of such differential equations is Michaelis-Menten kinetics (\citealp{michaelis2011original}). 
Consider the following Michaelis-Menten kinetics between three different kinds of enzymes $A$, $B$, $C$, and the input $I$:
\begin{equation}\label{enzyme-example}
\begin{aligned}
\frac{dA}{dt} &= k_{IA}I\frac{(1-A)}{(1-A)+K_{IA}} - F_{A}k'_{F_AA}\frac{A}{A+K'_{F_AA}}\\
\frac{dB}{dt} &= Ck_{CB}\frac{(1-B)}{(1-B)+K_{CB}} - F_B k'_{F_BB}\frac{B}{B+K'_{F_BB}}\\
\frac{dC}{dt} &= A k_{AC}\frac{(1-C)}{(1-C)+K_{AC}} - Bk'_{BC}\frac{C}{C+K'_{BC}}.
\end{aligned}
\end{equation}
Assume that the exact values of $k_{IA}$ and $k_{CB}$ are unknown. 
One way to proceed to study the model is to sample $k_{IA}$ and $k_{CB}$ randomly from a plausible range, say $\pa \triangleq [d_1, u_1]\times[d_2, u_2]$, and see if the model can exhibit the desired behavior of the enzymatic system. 
A typical approach in systems biology is to sample a number of parameters uniformly from $R$ and observe what kind of model behaviors are exhibited at the selected design points
(\citealp{ma2009defining}).
However, the change of dynamics w.r.t.\ the change of the values of $k_{IA}$ and $k_{CB}$ might be highly non-linear so that the observation based on insufficient number of uniform samples can be misleading. 
Suppose that we are interested in the adaptive behavior of the model. 
Adaptation refers to the ability of the system to respond (i.e., change the output level) to an input stimulus (i.e., change in input level), and then return to its original output level even when the change in the input level persists.
The adaptive behavior can be summarized as the \emph{sensitivity} and the \emph{precision} of the system. 
In the context of our example \eqref{enzyme-example}, the sensitivity is defined as the ratio $\left|\frac{(C_{\text{peak}}-C_0)/C_0}{(I_1-I_0)/I_0}\right|$ between the size of the response of the output $C$ and the size of the stimulus (i.e., the change in the input $I$), where $I_0$ and $C_0$ are the initial input and output levels respectively, $I_1$ is the new input level, and $C_{\text{peak}}$ is the maximum of the output level after the intput level changes from $I_0$ to $I_1$; 
the precision is defined as the ratio $\left|\frac{(I_1-I_0)/I_0}{(C_1-C_0)/C_0}\right|$ between the (long term) change in the input and output levels, where $C_1$ is the final output level of the system by the time the system stabilizes after the initial change due to the stimulus. 
These output measures quantify how well the system detects the change in the input level, and how robust is the system to such change, respectively.  
For more thorough description of adaptation, sensitivity, and precision, see \cite{ma2009defining}. 
For our purpose, just note that ODE in \eqref{enzyme-example} defines a mapping from the parameter space $\pa$ to the output space $\R^2$, each coordinate of which represents the sensitivity and the precision, respectively. 
Figure~\ref{fig:sysbio} compares the observations based on uniform sampling on the parameter space and observations based on uniform samples on the output space.
More specifically, the right plot shows the observation based on sampling $(\log_{10}k_{IA} + 1)/2$ and $(\log_{10}k_{CB} + 1)/2$ uniformly from $R = [0.35, 0.88]\times[0,1]$, and the other two plots show the observations based on sampling uniformly from the output space via Algorithm~\ref{alg:iterative} with $q=0.1$, $k = 5$, $h = 0.03$, $b=\infty$, and $N = 1000$. 
We used Algorithm~\ref{alg:iterative} instead of Algorithm~\ref{alg:ideal} in this example, since the exact derivative of the mapping is not readily available. 
The parameters of the ODE (other than $k_{IA}$ and $k_{CB}$) were chosen as follows:
$$
\left(\begin{array}{c}
k'_FAA\\
k'_FBB\\   
k_AC\\     
k'_BC\\    
K_IA\\     
K'_FAA\\   
K_CB\\  
K'_FBB\\
K_AC\\  
K'_BC
\end{array}\right)
= 
\left(\begin{array}{c}
7.0437 \\
0.1364 \\    
3.0061 \\ 
0.8395 \\   
0.0183 \\
0.0016 \\     
0.0122 \\  
0.0032 \\    
0.0044 \\  
0.0742
\end{array}\right).
$$
One can see that by uniformly sampling on the parameter space one may end up wasting lots of design points to explore the lower left part of the model output space while almost missing the protruding region on the lower right part of the model output space. On the other hand, our algorithm distributes the design points intelligently so that the nearly missed lower right part of the model output space is clearly identified with less total number (4000 times) of ODE simulations compared to the na\"ive design points uniform in the parameter space (5000 times). 
\begin{figure}
\hspace*{-120pt}\includegraphics[width=23cm, height=15cm]{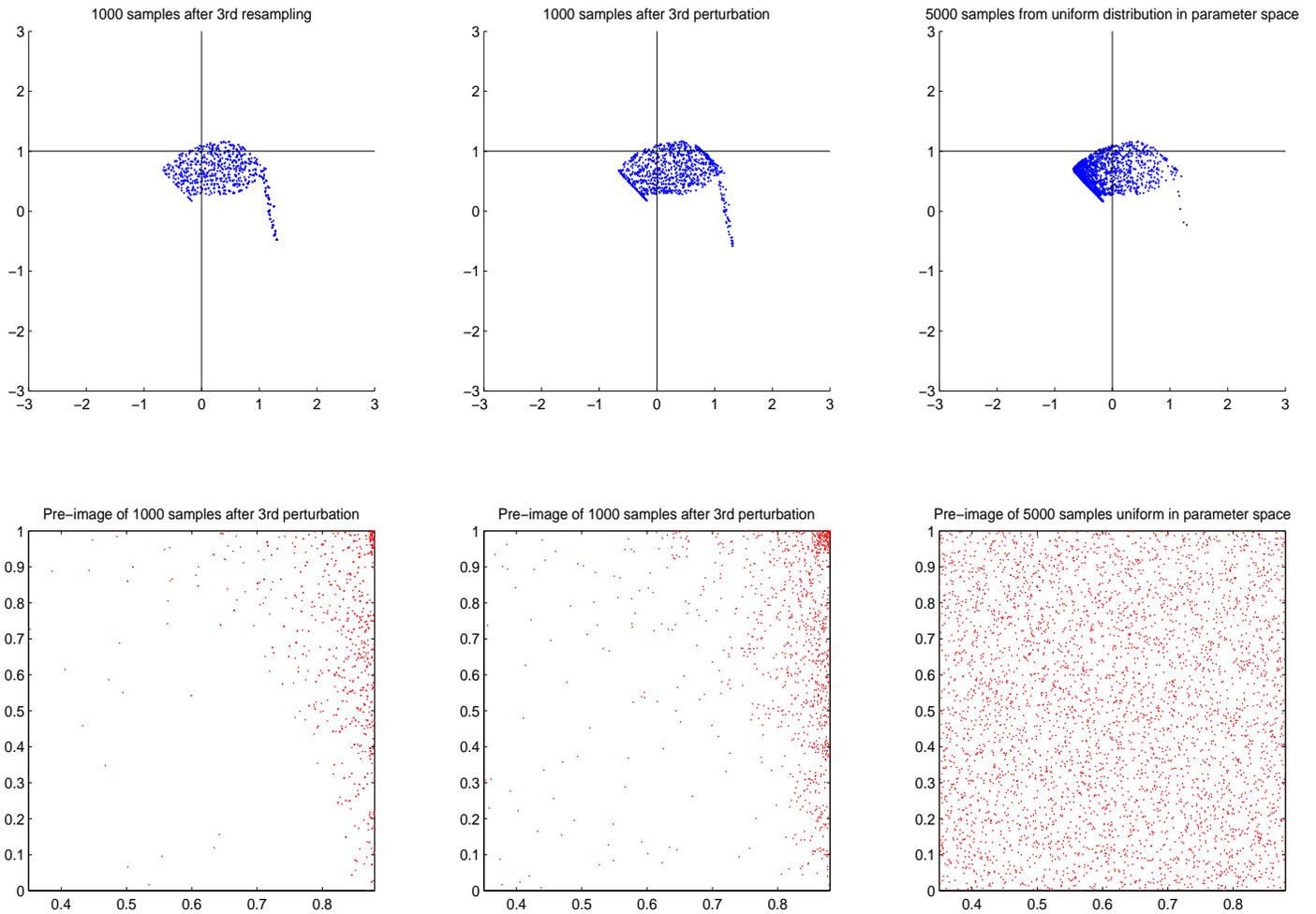}
\caption{The right plot shows the observation based on the samples uniformly distributed on $\pa = [0.35, 0.88]\times[0,1]$, and the other two plots show the observations based on Algorithm~\ref{alg:iterative}. \label{fig:sysbio}}
\end{figure}
\end{example}

\appendix


%
%
\section{Appendix}\label{sec:knn-proofs}
In this section we provide a justification for the consistency of Algorithm~\ref{alg:iterative}.
More specifically, let 
$$
\check\xi^{[j+1]}
\triangleq \sum_{i=1}^{N} \frac{(\mu\circ f\cdot \hat r ^m\circ f )(X_i^{[j+1/2]})}{\sum_{l=1}^{N} (\mu\circ f\cdot \hat r^m\circ f)(X_l^{[j+1/2]})}\delta_{X_i^{[j+1/2]}}.
$$
and consider a procedure that follows the same steps 1)-4) in Section~\ref{sec:proofs} except that 
\begin{itemize}
\item In step 1), set, instead of \eqref{eq:xi-j-plus-1}, 
$$
\xi^{[j+1/2]} 
\triangleq
\frac{\min\left\{b, \ q/\lambda^m(A) + (1-q)\frac{1}{N}\sum_{i=1}^{N}\zeta_h(x-x_i')\right\}}
{\int_A \min\left\{b, \ q/\lambda^m(A) + (1-q)\frac{1}{N}\sum_{i=1}^{N}\zeta_h(s-x_i')\right\}ds}
$$
\item In step 4), generate samples from $\check \xi^{[j+1]}$ instead of $\xi^{[j+1]}$.
\end{itemize}
Then, $\hat \xi^{[j+1]}$ describes the samples after resampling step (in $j+1$\textsuperscript{th} iteration) produced by Algorithm~\ref{alg:iterative}.
If we set $H = (\mu\circ f) \cdot \left(\frac{\Gamma_m\hat r^m\circ f}{k/N}\right)$, $G = d\xi/d\xi^{[j+1/2]}$, $\nu = \hat\xi^{[j+1/2]}$, 
then
$\Psi_H(\nu) = \check\xi^{[j+1]}$,
$\Psi_G(\nu) = \xi^{[j+1]}$, 
and hence,
\begin{equation*}
\begin{aligned}
\W_1(\hat\xi^{[j+1]}, \xi)
&
\leq 
\W_1(\hat \xi^{[j+1]}, \check\xi^{[j+1]})+
\W_1(\check\xi^{[j+1]}, \xi^{[j+1]}) + 
\W_1(\xi^{[j+1]}, \xi)
\\
&
=
\W_1(\hat \xi^{[j+1]}, \check\xi^{[j+1]})+
\W_1(\Psi_H(\nu), \Psi_G(\nu)) + \W_1(\xi^{[j+1]}, \xi)
.
\end{aligned}
\end{equation*}
Note that $\W_1(\hat \xi^{[j+1]}, \check\xi^{[j+1]})\leq cD\alpha(N)$ from \cite{fournier2015rate}, and 
one can show that $\E\W_1(\xi^{[j+1]}, \xi)$ can be bounded by $c\alpha(N)\E \W_1(\xi^{[j]}, \xi) + d\alpha(N)$ for some $c$ and $d$ following a similar argument as in Proposition~\ref{prop:key_prop}.
Since $H$ can be viewed as an approximation of $G$, for Algorithm~\ref{alg:iterative}'s consistency, what is left is to show that $\W_1(\check\xi^{[j+1]}, \xi^{[j+1]})=\W_1(\Psi_H(\nu), \Psi_G(\nu))$ can also be bounded in a similar form by establishing some sort of modulus of continuity of $\Psi_\cdot(\nu)$ in terms of $\W_1$ distance w.r.t.\ the potential. 
%
%
Note first that from a straightforward algebra,
$$
\big(\Psi_G(\nu) - \Psi_{H}(\nu)\big)f
= \frac{1}{\nu(G)} \big\{\nu\big((G-H)f\big) + \nu (H-G) \Psi_{H}(\nu)f \big\}
$$
and hence,
\begin{align}
\W_1(\Psi_G(\nu),\Psi_{H}(\nu))
&=\sup_{f: K_f \leq 1, f(x^0) = 0}\big\{\big(\Psi_G(\nu) - \Psi_{H}(\nu)\big)f\big\}
\nonumber
\\
&= \frac{1}{\nu(G)} \sup_{f: K_f \leq 1, , f(x^0) = 0}\big\{\nu\big((G-H)f\big) + \nu (H-G) \Psi_{H}(\nu)f \big\}\\
&\leq \frac{1}{\nu(G)} \sup_{f: K_f \leq 1, f(x^0) = 0}\big\{\nu\big(|H-G|\,\|f\|_\infty\big) + |\nu(H-G)| \,\|f\|_\infty\big\}
\nonumber
\\
&\leq \frac{1}{\nu(G)} \left\{\nu\big(|H-G|\big)\,\diam(\pa) + |\nu(H-G)|\,\diam(\pa)\right\}
\nonumber
\\
&\leq \frac{2\,\diam(\pa)}{\nu(G)} \nu(|H-G|).
\label{eq:mod_con_wrt_pot}
\end{align}
%
%
%
%
%
%
Therefore,
\begin{equation*}
\begin{aligned}
\W_1(\check\xi^{[j+1]}, \xi^{[j+1]})
&
\leq 
\frac{2\,\diam(\pa)}{\hat\xi^{[j+1/2]}(d\xi/d\xi^{[j+1/2]})}\,\hat\xi^{[j+1/2]}\left(\left|\mu\circ f\cdot \frac{\Gamma_m}{k/N}\hat r^m\circ f- d\xi / d\xi^{[j+1/2]}\right|\right) 
\\
&
\leq 
2\,b\,\diam(\pa)\,
\hat\xi^{[j+1/2]}\left(\left|\mu\circ f\cdot \frac{\Gamma_m}{k/N}\hat r^m\circ f- d\xi / d\xi^{[j+1/2]}\right|\right).
\end{aligned}
\end{equation*}
where the second inequality is from the construction of $\xi^{[j+1/2]}$.
Lemma~\ref{lemma:key_lemma} provides the desired bound for the RHS. 
For Lemma~\ref{lemma:key_lemma}, we make a few additional assumptions. 
\begin{itemize}
\item[A3.] The target density $\mu$ is bounded away from above and below;
\item[A4.] Spectrum of $Df$ is bounded away from $0$ and $\pm \infty$;
\item[A5.] There exists a constant $c_m>0$ and $\delta_0>0$ such that if $y_0 \in \mathcal M$
\begin{equation}\label{eq:key_lemma_key_assumption}
\int_{B(y_0; \delta_1)\setminus B(y_0; \delta_2)} \H^m(dy) \geq c_m \,(\delta_1^m - \delta_2^m)
\end{equation} 
for $\delta_0 \geq \delta_1 \geq \delta_2 \geq 0$.
\end{itemize}

\begin{lemma}\label{lemma:key_lemma}
For any given $\epsilon>0$, one can choose $k$ as a function of $N$ so that there exists $c$ such that
\begin{equation}\label{eq:key_lemma_conclusion1}
\E\hat\xi_N^{[j+1/2]}\left(\left|\mu\circ f\cdot \frac{\Gamma_m}{k/N}\hat r^m\circ f- d\xi / d\xi^{[j+1/2]}\right|\right)\leq  c\big(1+\|\nabla \zeta_h\|_\lip\W_1( \xi^{[j]},\xi)\big)N^{-\frac{1-\epsilon}{2(m+1)}}.
\end{equation}

\end{lemma}
\begin{proof}
First note that from (\ref{eq:density_of_xi}),
\begin{align*}
\left|\mu\circ f\cdot \frac{\Gamma_m}{k/N} \hat r^m\circ f- d\xi / d\xi^{[j+1/2]}\right|
&
= |\mu\circ f| \cdot \left|\frac{\Gamma_m}{k/N}\hat r^m\circ f - \frac{J_m f }{ d\xi^{[j+1/2]}/d\lambda^m}\right|
\\
&
= |\mu\circ f| \cdot \left|\frac{\Gamma_m}{k/N}\hat r^m\circ f - \frac{1}{ p_Y\circ f}\right|
\end{align*}
where $p_Y$ is the density of $f(X_i^{[j+1/2]})$'s w.r.t.\ the Hausdorff measure. Now we consider the following decomposition:
\begin{align}\label{eq:decomposition}
|\mu\circ f| \cdot\left|\frac{\Gamma_m}{k/N}\hat r^m\circ f - \frac{1}{p_Y\circ f}\right|
&\leq 
|\mu\circ f| \cdot\left|\frac{\Gamma_m}{k/N}\hat r^m\circ f - \frac{\Gamma_m}{k/N}r_{k}^m\circ f \right|
+
|\mu\circ f| \cdot\left|\frac{\Gamma_m}{k/N} r_{k}^m\circ f - \frac{1}{p_Y\circ f}\right| \nonumber\\
&= (I) + (II)
\end{align}
where for each $y_0$, $r_{k}(y_0)$ is the real number such that
$$\int_{B(y_0; r_k(y_0))}p_Y(y)\H^m(dy) = k/N.$$
For (I),
\begin{align*}
&
\E\frac{\hat\xi_N^{[j+1/2]} (|\mu\circ f|\cdot|\hat r^m \circ f- r_k^m\circ f|)}{k/N} 
= \E \frac{\left|\mu(f(X_i^{[j+1/2]}))\right|\cdot\left|\hat r^m(f(X_1^{[j+1/2]})) - r_k^m(f(X_1^{[j+1/2]}))\right|}{k/N}
\\
&
= \E\left[ \E \left[\left.\frac{\left|\mu(f(X_1^{[j+1/2]}))\right|\cdot\left|\hat r^m(f(X_1^{[j+1/2]})) - r_k^m(f(X_1^{[j+1/2]}))\right|}{k/N}\right|f(X_1^{[j+1/2]})\right]\right].
\end{align*}
We first study the inner conditional expectation given $f(X_1^{[j+1/2]}) = y_0$, which will be denoted by $\E_{y_0}$ from now on, i.e.,
$$
\E_{y_0} \frac{\left|\mu({y_0})\right|\cdot\left|\hat r^m(y_0) - r_k^m({y_0})\right|}{k/N}
\triangleq
 \E \left[\left.\frac{\left|\mu(f(X_1^{[j+1/2]}))\right|\cdot\left|\hat r^m(f(X_1^{[j+1/2]})) - r_k^m(f(X_1^{[j+1/2]}))\right|}{k/N}\right|f(X_1^{[j+1/2]})={y_0}\right].
$$
With this notation,
\begin{align}
&\E_{y_0} \frac{\left|\mu({y_0})\right|\cdot\left|\hat r^m({y_0}) - r_k^m({y_0})\right|}{k/N}
= \int_0^\infty \P_{y_0}\left(\frac{|\mu({y_0})|\cdot\left|\hat r^m({y_0}) - r_k^m({y_0})\right|}{k/N} \geq s\right)ds
\nonumber\\
&\leq \gamma + \int_\gamma^\infty \P_{y_0}\left(\frac{|\mu({y_0})|\cdot\left|\hat r^m({y_0}) - r_k^m({y_0})\right|}{k/N} \geq s\right)ds
\nonumber\\
&\leq \gamma + \int_\gamma^\infty\P_{y_0}\left(\hat r^m({y_0}) - r_k^m({y_0}) \geq \frac{ks}{N\mu({y_0})}\right)ds
+ \int_\gamma^{\infty}\P_{y_0}\left(r_k^m({y_0}) - \hat r^m({y_0}) \geq \frac{ks}{N\mu({y_0})}\right)ds 
\label{eq:bound_expectation}
\end{align}
where $\P_{y_0}$ is the corresponding conditional probability given $f(X_1^{[j+1/2]}) = {y_0}$.
Note that since $\M$ is bounded, the upper limit of the first integral is in fact finite:
\begin{align*}
&\int_\gamma^\infty\P_{y_0}\left(\hat r^m({y_0}) - r_k^m({y_0}) \geq \frac{ks}{N\mu({y_0})}\right)ds
=\int_\gamma^{(N/k)\diam(\M)}\P_{y_0}\left(\hat r^m({y_0}) - r_k^m({y_0}) \geq \frac{ks}{N\mu({y_0})}\right)ds.
\end{align*}
Also,
\begin{align}
&\P_{y_0}\left(\hat r^m({y_0}) - r_k^m({y_0}) \geq \frac{ks}{N\mu(y_0)}\right) 
\nonumber
\\
&
= \P_{y_0}\left(\text{less than $k$ points fall inside $B\Big({y_0}; (r_k^m({y_0}) + \mu(y_0)^{-1}ks/N)^{1/m}\Big)$}\right)
\nonumber
\\
&
= \P(\text{Bin}(N, q) \leq k)
\label{eq:binomial}
\end{align}
where
$q = q(s,y_0) \triangleq \int_{B(y_0; (r_k^m(y_0)+\mu(y_0)^{-1}ks/N)^{1/m})}  \frac{\frac{d\xi_N^{[j+1/2]}}{d\lambda^m}\circ f^{-1}(y)}{J_mf \circ f^{-1}(y)} \H^m(dy)$. Note that
\begin{align*}
q-k/N
&= \int_{B(y_0; (r_k^m(y_0)+\mu(y_0)^{-1}kx/N)^{1/m})}  \frac{\frac{d\xi_N^{[j+1/2]}}{d\lambda^m}\circ f^{-1}(y)}{J_mf \circ f^{-1}(y)} \H^m(dy) - k/N\\
&= \int_{B(y_0; (r_k^m(y_0)+\mu(y_0)^{-1}ks/N)^{1/m})}  p_Y(y) \H^m(dy) - \int_{B(y_0; r_k(y_0))} p_Y(y) \H^m(dy)\\
&= \int_{B(y_0; (r_k^m(y_0)+\mu(y_0)^{-1}ks/N)^{1/m})\setminus B(y_0; r_k(y_0))} p_Y(y) \H^m(dy)\\
&\geq \frac{csk}{N}\cdot\frac{p_Y(y_0+h')}{\mu(y_0)}
\end{align*}
where $y_0+h' \in B(y_0; (r_k^m(y_0)+\mu(y_0)^{-1}ks/N)^{1/m})$ by \eqref{eq:key_lemma_key_assumption} and the mean-value theorem as far as $(\frac{ks}{N\mu(y_0)} + r_k^m)^{1/m} \leq \delta_0$.
Therefore, from Hoeffding's inequality and (\ref{eq:binomial}),
\begin{align*}
\P_{y_0}\left(\hat r^m(y_0) - r_k^m(y_0) \geq \frac{ks}{N\mu(y_0)}\right)
\leq \exp \left(-2N\left(\frac{cks\,p_Y(y_0+h')}{N\mu(y_0)}\right)^2\right)
\end{align*}
for $s \leq (N/k)\mu(y_0)(\delta_0^m - r_k^m(y_0))$. 
In view of this,
\begin{align*}
&
\int_\gamma^{(N/k)\mu(y_0)\diam(\M)^m}\P_{y_0}\left(\hat r^m(y_0) - r_k^m(y_0) \geq \frac{ks}{N\mu(y_0)}\right)ds
\\
&
\leq \int_\gamma^{(N/k)\mu(y_0)\diam(\M)^m}\P_{y_0}\left(\hat r^m(y_0) - r_k^m(y_0) \geq \frac{k\gamma}{N\mu(y_0)}\right)ds
\\
&
\leq (N/k)\mu(y_0)\diam(\M)^m\exp\left(-\frac{2c^2k^2\gamma^2\,p_Y(y_0+h')^2}{N\mu(y_0)^2}\right)
\end{align*}
if $\gamma \leq (N/k)(\delta_0^m - r_k^m(y_0))$. 
Similarly, the second integral in (\ref{eq:bound_expectation}) can be bounded by first noting that $\hat r$ is non-negative (and hence the upper limit of the integral is finite), and then applying Hoeffding's inequality:
\begin{align*}
\int_\gamma^{\infty}\P_{y_0}\left(r_k^m(y_0) - \hat r^m(y_0) \geq \frac{ks}{N\mu(y_0)}\right)ds
&
\leq 
\int_\gamma^{(N/k)\mu(y_0)r_k(y_0)^m}\P_{y_0}\left(r_k^m(y_0) - \hat r^m(y_0) \geq  \frac{k\gamma}{N\mu(y_0)} \right)ds\\
&\leq
(N/k)\mu(y_0){r_k}^m(y_0)\exp\left(-\frac{2c^2k^2\gamma^2\,p_Y(y_0+h')^2}{N\mu(y_0)^2}\right)
\end{align*}
for $\gamma \leq (N/k)(\delta_0^m - {r_k}^m(y_0))$.
From these along with (\ref{eq:bound_expectation}), we conclude that
\begin{align*}
\E_{y_0}\frac{|\mu(y_0)|\cdot\left|\hat r^m(y_0) - r_k^m(y_0)\right|}{k/N} 
&\leq \gamma + (N/k)\mu(y_0)\big(\diam(\M)^m+ r_k^m(y_0)\big)\exp\left(-\frac{2c^2k^2\gamma^2\,p_Y(y_0+h')^2}{N\mu(y_0)^2}\right)
\nonumber\\
&\leq \gamma + 2(N/k)\mu(y_0)\diam(\M)^m\exp\left(-\frac{2c^2k^2\gamma^2\,p_Y(y_0+h')^2}{N\mu(y_0)^2}\right).
\end{align*}
for $\gamma \leq (N/k)(\delta_0^m - {r_k}^m(y_0))$.
Choosing $\gamma = \frac{N^{1/2+\beta}}{k}\frac{\mu(y_0)}{p_Y(y_0)}$ and $k = N^{\alpha}$ for $\alpha \in (1/2,1)$ and $\beta \in (0,1/2)$, we get
$$
|h'|
\leq 
\left(r_k^m(y_0) + \frac{k\gamma/N}{\mu(y_0)}\right)^{1/m}
=\left(r_k^m(y_0) + \frac{N^{\beta-1/2}}{p_Y(y_0)}\right)^{1/m}
$$
and
\begin{align}
&
\E_{y_0}\frac{|\mu(y_0)|\cdot\left|\hat r^m(y_0) - r_k^m(y_0)\right|}{k/N} 
\nonumber
\\
&
\leq 
N^{1/2+\beta-\alpha}\frac{\mu(y_0)}{p_Y(y_0)}
+ 2N^{1-\alpha}\mu(y_0)\diam(\M)^m\exp\left(-2c^2N^{2\beta}\left(\frac{p_Y(y_0+h')}{ p_Y(y_0)}\right)^2\right)
\nonumber
\\
&
\leq 
N^{1/2+\beta-\alpha}\frac{\mu(y_0)}{p_Y(y_0)}
+ 2N^{1-\alpha}\mu(y_0)\diam(\M)^m\exp\left(-2c^2N^{2\beta}\left(\underline{p_Y}/ \overline{p_Y}\right)^2\right).
\label{eq:bound-for-I}
\end{align}
Therefore, 
\begin{align}\label{bound:xi_hat_r_hat_r}
\E \frac{\hat \xi_N^{[j+1/2]}(|\mu\circ f|\cdot|\Gamma_m\hat r^m\circ f - \Gamma_m r_k^m\circ f|)}{k/N} 
\leq 
c \{\underline{p_Y}^{-1}N^{1/2+\beta-\alpha}+N^{1-\alpha}\exp(-cN^{2\beta}(\underline{p_Y}/\overline{p_Y})^2\}
\end{align}
for some $c>0$ that depends only on $f$, $\mu$, and $\diam(M)$.
\cf{
\begin{align*}
&
\E_{y_0}\frac{|\mu(y_0)|\cdot\left|\hat r^m(y_0) - r_k^m(y_0)\right|}{k/N} 
\\
&
\leq 
N^{1/2+\beta-\alpha}\frac{\mu(y_0)}{p_Y(y_0)}
+ 2N^{1-\alpha}\mu(y_0)\diam(\M)^m\exp\left(-2c^2N^{2\beta}\left(\frac{p_Y(y_0+h')}{ p_Y(y_0)}\right)^2\right)
\\
&
\leq 
N^{1/2+\beta-\alpha}\frac{\mu(y_0)}{p_Y(y_0)}
+ 2N^{1-\alpha}\mu(y_0)\diam(\M)^m\exp\left(-2c^2N^{2\beta}\left(1+\frac{K_{p_Y}}{p_Y(y_0)}\left(r_k^m(y_0) + \frac{N^{\beta-1/2}}{p_Y(y_0)}\right)^{1/m}\right)^2\right).
\end{align*}
}

Now, turning to (II) of (\ref{eq:decomposition}), we first prove that for $y_0$ and $\delta$ such that $\bar f^{-1}(B(y_0; \delta)) \subseteq A$ (where $\bar f$ is defined in (\ref{eq:linear_approximation})),
\begin{equation}\label{eq:mean_value_theorem_H}
\left|\frac{1}{\Gamma_m\delta^m}\int_{B(y_0;\delta)} \H^m(dy) - 1\right|
\leq c_f' \delta
\end{equation}
for $c_f'$ that depends only on $f$.
Let $\bar f$ be the linear approximation of $f$ at $x_0 = f^{-1}(y_0)$, i.e.,
\begin{equation}\label{eq:linear_approximation}
\bar f(x_0+h) =  f(x_0)+ Df(x_0)h,
\end{equation}
and $\bar\H^m$ be the Hausdorff measure on $\bar f(\pa)$. Then,
\begin{align*}
\int_{B(y_0;\delta)} \H^m(dy) - \Gamma_m \delta^m
&= \int_{B(y_0;\delta)} \H^m(dy) - \int_{B(y_0; \delta)} \bar{\H}^m(dy)\\
&= \int_{f^{-1}(B(y_0;\delta))} J_mf(x)\lambda^m(dx) - \int_{\bar{f}^{-1}(B(y_0; \delta))} J_m\bar f(x){\lambda}^m(dx)\\
&= \int_{f^{-1}(B(y_0;\delta))\cap \bar f^{-1}(B(y_0;\delta))} \big(J_mf(x) - J_m \bar f(x)\big)\lambda^m(dx) \\
&\qquad + \int_{f^{-1}(B(y_0;\delta))\backslash \bar f^{-1}(B(y_0;\delta))} J_m f(x){\lambda}^m(dx)\\ 
&\qquad - \int_{\bar f^{-1}(B(y_0;\delta))\backslash  f^{-1}(B(y_0;\delta))} J_m \bar f(x){\lambda}^m(dx)\\
&= \text{(I)}' + \text{(II)}' - \text{(III)}'
\end{align*}
Since $J_m \bar f(x) = J_m f(x_0)$, the first term (I)$'$ in the previous display can be bounded as follows:
\begin{align*}
|\text{(I)}'| &= \int_{f^{-1}(B(y_0;\delta))\cap \bar f^{-1}(B(y_0;\delta))} \big|J_mf(x) - J_m f(x_0)\big|\lambda^m(dx) \\
&\leq 
\int_{\bar f^{-1}(B(y_0;\delta))} J_m \bar f(x) \lambda^m(dx) \ \sup_{x\in \bar f^{-1}(B(y_0; \delta))}\left|1 - \frac{J_m f(x)}{J_mf(x_0)}\right|\\
&\leq 
\Gamma_m \delta^m \cdot \sup_{x\in \bar f^{-1}(B(y_0; \delta))}\left|1 - \frac{J_m f(x)}{J_mf(x_0)}\right|
\end{align*}
To get a bound for (II)$'$, note that
\begin{align*}
|\text{(II)}'| 
&\leq\int_{f^{-1}(B(y_0;\delta))\backslash \bar f^{-1}(B(y_0;\delta))} J_m f(x_0) {\lambda}^m(dx)\cdot  \sup_{x\in \bar f^{-1}(B(y_0; \delta))}\left|\frac{J_m f(x)}{J_mf(x_0)}\right|\\ 
&= \int_{f\big(f^{-1}(B(y_0; \delta))\backslash \bar f^{-1}(B(y_0; \delta))\big)}\bar \H^m(dy)
\cdot  \sup_{x\in \bar f^{-1}(B(y_0; \delta))}\left|\frac{J_m f(x)}{J_mf(x_0)}\right|
\end{align*}
To bound the integral, we prove that $f\big(f^{-1}(B(y_0; \delta))\backslash \bar f^{-1}(B(y_0; \delta))\big)$ is close to the boundary of $B(y_0; \delta)$---i.e., $f\big(f^{-1}(B(y_0; \delta))\backslash \bar f^{-1}(B(y_0; \delta))\big)\subset B(y_0; \delta) \backslash B(y_0;\delta-\epsilon(\delta))$ where $\epsilon(\delta) = o(\delta)$ as $\delta\to0$. Suppose that $y \in f\big(f^{-1}(B(y_0; \delta))\backslash \bar f^{-1}(B(y_0; \delta))\big)$. Then, there exists an $h\in \R^m$ such that $x_0+h = f^{-1}(y)$,  $f(x_0+h) \in B(y_0; \delta)$, and $\bar f(x_0+h) \notin B(y_0; \delta)$. 

Since $f$ is $C^2$, the $k$\textsuperscript{th} component of $f$ can be written as
$$
f_k(x_0+h) = f_k(x_0)+ Df_k(x_0)h + \frac{1}{2}h^TR_k(x_0+h)h
$$
where $R_k(x_0+h) = \int_0^1 (1-t)D^2f_k(x_0+th)dt$. Therefore, $\|R_k(x_0+h)\|\leq \sup_{t\in[0,1]}\| D^2 f_k(x_0+th)\|$, 
\begin{equation}\label{eq:f-barf_bound1}
\|f(x_0+h) - \bar f(x_0+h)\|_2 \leq \frac{1}{2} \sum_{k=1}^n |h^TR_k(x_0+h)h| \leq \frac{1}{2}\sup_{x\in \pa}\| D^2 f_k(x)\| \|h\|_2^2.
\end{equation}
Assumption A4 guarantees that $\|f^{-1}\|_\lip>0$ and $\text{diam}(f^{-1}(B(y_0; \delta))<2\,\|f^{-1}\|_\lip\delta$, and since $x_0+h \in f^{-1}(B(y_0;\delta))$---i.e., $\|h\|_2\leq 2\,\|{f^{-1}}\|_\lip\delta$---(\ref{eq:f-barf_bound1}) becomes
$$\|f(x_0+h) - \bar f(x_0+h)\|_2 \leq 2\, \|{f^{-1}}\|_\lip^2\,\sup_{x\in \pa}\| D^2 f_k(x)\| \delta^2 \triangleq c_f\, \delta^2.$$
Since $\tilde f(x_0+h) \notin B(y_0,\delta)$, the above inequality implies that $f(x_0+h) \in B(y_0; \delta)\backslash B(y_0; \delta-c_f\delta^2)$, and hence, $f\big(f^{-1}(B(y_0; \delta))\backslash \bar f^{-1}(B(y_0; \delta))\big)$ is a subset of $B(y_0; \delta)\backslash B(y_0; \delta-c_f\delta^2)$. 
Now, we can bound the integral in (II)$'$
$$ \int_{f\big(f^{-1}(B(y_0; \delta))\backslash \bar f^{-1}(B(y_0; \delta))\big)} \bar\H^m(dy) \leq  \Gamma_m(\delta^m - (\delta-c_f\delta^2)^m) \leq m\Gamma_m c_f \delta^{m+1},$$
which in turn implies 
\begin{equation*}
\text{(II)}' < m\Gamma_m c_f \delta^{m+1}\cdot  \sup_{x\in \bar f^{-1}(B(y_0; \delta))}\left|\frac{J_m f(x)}{J_mf(x_0)}\right|.
\end{equation*}
The following bound for (III)$'$ can be obtained by an essentially identical (but only simpler) argument: 
$$
\text{(III)}' = \int_{\bar f^{-1}(B(y_0;\delta))\backslash  f^{-1}(B(y_0;\delta))} J_m \bar f(x){\lambda}^m(dx) 
\leq m\Gamma_m c_f \delta^{m+1}
$$
Now, combining the bounds for (I)$'$, (II)$'$, and (III)$'$, we arrive at 
\begin{equation}\label{eq:mvt_for_Hausdorff}
\left|\frac{1}{\Gamma_m\delta^m}\int_{B(y_0;\delta)} \H^m(dy) - 1\right|
\leq \sup_{x\in \bar f^{-1}(B(y_0; \delta))}\left|1 - \frac{J_m f(x)}{J_mf(x_0)}\right| + m c_f\delta\left(\sup_{x\in \bar f^{-1}(B(y_0; \delta))}\left|\frac{J_m f(x)}{J_mf(x_0)}\right| + 1\right) 
\end{equation}
and
$$
\sup_{x\in \bar f^{-1}(B(y_0; \delta))}\left|1 - \frac{J_m f(x)}{J_mf(x_0)}\right| \leq \delta \|J_m f\|_\lip\, \|{f^{-1}}\|_\lip/\underline{J_mf}
$$
where $\underline{J_m f} = \inf_{x\in \pa} J_m f(x)$. 
Letting $c_f' = \|{J_m f}\|_\lip \|{f^{-1}}\|_\lip/\underline{J_mf} + m c_f \left(\sup_{x\in \bar f^{-1}(B(y_0; \delta))}\left|\frac{J_m f(x)}{J_mf(x_0)}\right| + 1\right)$, we arrive at (\ref{eq:mean_value_theorem_H}).
%
Note that from the mean value theorem,
\begin{align}
\int_{B(y_0;\delta)}p_Y(y)\H^m(dy) 
&= \int_{f^{-1}(B(y_0; \delta))} J_mf(x) p_Y(f(x))\lambda^m(dx) \nonumber \\
&= \int_{f^{-1}(B(f(x_0); \delta))} J_m f(x) \lambda^m(dx)  \ p_Y\circ f(x_0+h^*) \nonumber \\
&= \int_{B(y_0; \delta)} \H^m(dy)\ p_Y\circ f(x_0+h^*)
\label{eq:intermediate1}
\end{align}
for some $h^*$ such that $f(x_0+h^*) \in B(y_0; \delta)$ where $x_0 = f^{-1}(y_0)$. 
Substituting $r_k$ for $\delta$ in (\ref{eq:intermediate1}) and using the definition of $r_k(y_0)$, we get
\begin{align}\label{eq:2.12}
\frac{\Gamma_m r_k^m\circ f(x_0)}{k/N} 
&= 
\frac{\Gamma_m r_k^m\circ f(x_0)}{\int_{B(f(x_0);r_k\circ f(x_0))} \H^m(dy)}
\frac{p_Y\circ f(x_0)}{p_Y\circ f(x_0+h^*)}
\frac{1}{p_Y\circ f(x_0)}.
\end{align}
\ccf{
$$|bc-1| \leq  |b-1|\cdot|c| + |c-1|$$
\begin{align*}
&\left|
\frac{1}{1 + \epsilon_3(r_k(y_0))}
\frac{p_Y(y_0)}{p_Y(y_0)+\epsilon_4(r_k(y_0))}-1\right|\\
&\leq
\frac{c'_f r_k(y_0)}{1 - c'_f\cdot r_k(y_0)}\frac{p_Y(y_0)}{p_Y(y_0)-K_{p_Y}\cdot r_k(y_0)} + \frac{K_{p_Y}\cdot r_k(y_0)}{p_Y(y_0)-K_{p_Y}\cdot r_k(y_0)}\\
&\leq
\frac{c'_f r_k(y_0)}{1 - c'_f\cdot r_k(y_0)}\frac{\underline{p_Y}}{\underline{p_Y}-K_{p_Y}\cdot r_k(y_0)} + \frac{K_{p_Y}\cdot r_k(y_0)}{\underline{p_Y}-K_{p_Y}\cdot r_k(y_0)}
\end{align*}
}
We therefore arrive at the bound
\begin{align}
&
\left|\frac{1}{p_Y\circ f(x_0)}-\frac{\Gamma_mr_k^m\circ f(x_0)}{k/N}\right|
\\
&=
\frac{1}{p_Y(y_0)}
\left(
\frac{p_Y\circ f(x_0+h^*)-p_Y\circ f(x_0)}{p_Y\circ f(x_0+h^*)} + \frac{p_Y\circ f(x_0)}{p_Y\circ f(x_0+h^*)} \frac{\Gamma_m r_k^m(y_0)}{\int_{B(y_0; r_k(y_0))} \H^m(dy)}\left(\frac{\int_{B(y_0; r_k(y_0))}\H^m(dy)}{\Gamma_m r_k^m(y_0)}-1\right)
\right)
\end{align}
for $y_0$ s.t.\ $\bar f^{-1}(B(y_0; \delta)) \subseteq \pa$ and $h^*$ such that $f(x_0+h^*) \in B(y_0; \delta)$.
Note that if $x_0\in \pa^{-r_k\circ f(x_0)\|{f^{-1}}\|_\lip}$, then $\bar f^{-1} (B(f(x_0);r_k\circ f(x_0)) \subseteq \pa$.
\cf{$x'\in\bar f^{-1} (B(f(x_0);r_k\circ f(x_0)))$ implies 
$|\bar f(x') - f(x_0)| = |\bar f(x') - \bar f(x_0)| \leq r_k \circ f(x_0)$ which in turn implies $|x'-x_0| = |\bar f^{-1}(\bar f(x'))-\bar f^{-1}(\bar f(x_0))| \leq \|{f^{-1}}\|_\lip r_k \circ f(x_0) \leq \|{f^{-1}}\|_\lip\bar r_k$ and hence $x'\in \pa$
}
If, in addition, $x_0 \in \pa^{-2r_k\circ f(x_0)\|{f^{-1}}\|_\lip}$, then $x_0+h^* \in f^{-1}(B(f(x_0); r_k\circ f(x_0))$ implies $|h^*| \leq r_k\circ f(x_0) \|{f^{-1}}\|_\lip$, and hence, $x_0+h^* \in \pa^{-r_k\circ f(x_0) \|{f^{-1}}\|_\lip}$. 
For notational simplicity, let $\bar r_k \triangleq \|r_k\|_\infty$, $\underline{p_Y} \triangleq \|1/p_Y\|_\infty$, and $\overline{p_Y} \triangleq \|p_Y\|_\infty$.
Then, if $x_0 \in \pa^{-2\bar r_k \|{f^{-1}}\|_\lip}$, then assuming that $\bar r_k$ is sufficiently small,
\begin{align*}
\left|\frac{1}{p_Y\circ f(x_0)}-\frac{\Gamma_mr_k^m\circ f(x_0)}{k/N}\right|
&\leq
\frac{1}{p_Y(y_0)}
\Bigg(\frac{\|{p_Y\circ f}\|_\lip h^* }{p_Y\circ f(x_0+h^*)}+\frac{p_Y(y_0)}{p_Y(y_0+h^*)}\frac{c'_fr_k(y_0)}{1-c'_fr_k(y_0)}\Bigg)
\\
&\leq
\frac{r_k(y_0)}{p_Y(y_0)}
\Bigg(\frac{\|{f^{-1}}\|_\lip\, \|{p_Y\circ f}\|_\lip }{\underline{p_Y}}+\frac{\overline{p_Y}}{\underline{p_Y}}\frac{c'_f}{1-c'_fr_k(y_0)}\Bigg)
\\
&\leq 
\frac{c\bar r_k}{ p_Y(y_0)}\,
\Bigg(\|{p_Y\circ f}\|_\lip +1\Bigg).
\end{align*}
for some $c>0$ that depends only on $f$, $b$, and $q$.
From Lemma~\ref{lemma:new-l_infty_proximity} and the construction of $d\xi^{[j+1/2]}$,
\begin{align*}
\|{p_Y\circ f}\|_\lip
&
\leq 
\|{d\xi^{[j+1/2]}/d\lambda^m}\|_\lip\cdot \|{1}/{J_mf}\|_\infty + \|{1/J_mf}\|_\lip\, \|d\xi^{[j+1/2]}/d\lambda^m\|_\infty
\\
&
\leq
c\big(1+h+\|{\nabla \zeta_h}\|_\lip\W_1(\xi^{[j]}, \xi)\big).
\end{align*}
Therefore,
\begin{align*}
&
\E \hat\xi^{[j+1/2]}\left(|\mu\circ f|\cdot\left|\frac{\Gamma_mr_k^m\circ f}{k/N} - \frac{1}{p_Y\circ f}\right|\right)
\\
&
=
\E \mu(f(X_1^{[j+1/2]})) \left|\frac{1}{p_Y(f(X_1^{[j+1/2]}))} - \frac{\Gamma_m r_k^m(f(X_1^{[j+1/2]}))}{k/N}\right|
\\
&\qquad\qquad\qquad\qquad\qquad\qquad\qquad\qquad\cdot 
\left(
\I_{\{X_1^{[j+1/2]}\in \pa\setminus \pa^{-2\bar r_k \|{f^{-1}}\|_\lip}\}} + \I_{\{X_1^{[j+1/2]}\in \pa^{-2\bar r_k \|{f^{-1}}\|_\lip}\}}
\right)
\\
&
\leq  
\|\mu\|_\infty\Big(\frac{\Gamma_m \bar r_k ^m}{k/N} + \underline{p_Y}^{-1}\Big) \P\big(X_N^{[i]}\in \pa\setminus \pa^{-2\bar  r_k \|{f^{-1}}\|_\lip}\big)
\\
&\qquad+
\E \mu(f(X_1^{[j+1/2]})) \left|\frac{1}{p_Y(f(X_1^{[j+1/2]}))} - \frac{\Gamma_m r_k^m(f(X_1^{[j+1/2]}))}{k/N}\right|
\I_{\{X_1^{[j+1/2]}\in \pa^{-2\bar r_k\|{f^{-1}}\|_\lip}\}}
\\
&
\leq  
\|\mu\|_\infty\Big(\frac{\Gamma_m \bar r_k^m}{k/N} + \underline{p_Y}^{-1}\Big) \P(X_N^{[i]}\in \pa\setminus \pa^{-2\bar r_k \|{f^{-1}}\|_\lip})
+
c\E\frac{\bar r_k}{p_Y(f(X_1^{[j+1/2]}))}\Bigg(\|p_Y\circ f\|_\lip +1\Bigg)
\\
&
\leq
 c\big(1+h+\|\nabla \zeta_h\|_\lip \W_1( \xi^{[j]},\xi)\big)\bar r_k.
\end{align*}
for some (new) constant $c>0$.
Therefore, together with (\ref{bound:xi_hat_r_hat_r}) and noting that $\bar r_k ^m = O((k/N)) = O(N^{\alpha-1})$ (since $p_Y$ is bounded from below by construction),
\begin{align*}
&\E\hat\xi^{[j+1/2]}\left(\left|\mu\circ f\cdot \frac{\Gamma_m}{k/N}\hat r^m\circ f- d\xi / d\xi^{[j+1/2]}\right|\right)
\\
&\leq 
c\big\{ \big(1+ \|\nabla \zeta_h\|_\lip \W_1(\xi^{[j]},\xi)\big) N^{\frac{\alpha-1}{m}}
+
N^{1/2+\beta-\alpha}+N^{1-\alpha}\exp(-cN^{2\beta}(\underline{p_Y}/\overline{p_Y})^2\big\}
\end{align*}
Noting that $\underline{p_Y}/\overline{p_Y}$ is bounded from below by construction, we see that for any $\epsilon>0$, by considering $\beta$ small enough and $\alpha \approx \frac{m+2}{2(m+1)}$, one can always find a (new) constant $c>0$ such that 
$$
\E\hat\xi_N^{[j+1/2]}\left(\left|\mu\circ f\cdot \frac{\Gamma_m}{k/N}\hat r^m\circ f- d\xi / d\xi_N^{[j+1/2]}\right|\right)
\leq c\big(1+\|{\nabla \zeta_h}\|_\lip\W_1(\xi^{[j]},\xi)\big)N^{-\frac{1-\epsilon}{2(m+1)}}.
$$

\end{proof}


\bibliographystyle{apalike}
\bibliography{RZQ_ref}

\end{document}